\documentclass[%,
onecolumn]{IEEEtran}
\usepackage{amssymb}
\usepackage{amsmath}
\usepackage{enumerate}
\usepackage{amsthm}
\RequirePackage{algorithm}
\RequirePackage{algorithmic}
\makeatletter

\makeatletter
\let\NAT@parse\undefined
\makeatother
\usepackage[numbers]{natbib}

\makeatletter

\usepackage{verbatim}
\usepackage{setspace}
\usepackage{amsbsy}

\newtheorem{thm}{Theorem}
\newtheorem*{thm*}{} 
 
\newtheorem{lem}[thm]{Lemma}
\newtheorem*{lem*}{Lemma}
\newtheorem{prop}[thm]{Proposition}

\theoremstyle{definition}
 
\newtheorem{dfn}[thm]{Definition}

\theoremstyle{remark}

\newtheorem*{rem*}{Remark}

\newcommand{\mathd}{\mathrm{d}}
\newcommand{\mathe}{\mathrm{e}}

\newcommand{\EXP}{\ensuremath{\mathbb{E}}}

\newcommand{\IND}{\mathbb{I}}
\newcommand{\KT}{\textsc{kt}}
\newcommand{\RV}{\textsc{rv}}
\newcommand{\Np}{\ensuremath{\mathbb{N}_+}}
\newcommand{\N}{\ensuremath{\mathbb{N}}}

\newcommand{\cX}{\ensuremath{\mathcal{X}}}
\newcommand{\cP}{\ensuremath{\mathcal{P}}}

\newcommand{\bR}{{\operatorname{\sf R}}}
\newcommand{\cC}{\ensuremath{\mathcal{C}}}
\newcommand{\RED}{\ensuremath{\overline{R}}}
\newcommand{\REG}{\ensuremath{\widehat{R}}}
\newcommand{\PROB}{\ensuremath{\mathbb{P}}}

\newcommand{\cM}{\ensuremath{\mathcal{M}}}

\makeatother

\begin{document}
\title{Pattern Coding Meets Censoring: (almost)  
Adaptive Coding on Countable Alphabets}

\author{
\IEEEauthorblockN{Anna Ben-Hamou\IEEEauthorrefmark{1}\thanks{\IEEEauthorrefmark{1} Universit\'e Paris Diderot, Sorbonne Paris Cit\'e, CNRS UMR 7599 Laboratoire Probabilit\'es et Mod\`eles Al\'eatoires 75013 Paris} \and 
St\'ephane Boucheron\IEEEauthorrefmark{2}\thanks{\IEEEauthorrefmark{2} Universit\'e Paris Diderot, Sorbonne Paris Cit\'e, CNRS UMR 7599 Laboratoire Probabilit\'es et Mod\`eles Al\'eatoires  75013 Paris \& DMA Ecole Normale Sup\'erieure rue d'Ulm, 75005 Paris} \and Elisabeth Gassiat\IEEEauthorrefmark{3}\thanks{\IEEEauthorrefmark{3} Laboratoire de Math\'ematiques d'Orsay, Univ. Paris-Sud, CNRS,  Universit\'e Paris-Saclay 91405 Orsay}
}
}
\maketitle

\date{\today}
\maketitle

%%%%%%%%%%%%%%%%%%%%%%%%%%%%%%%%%%%%%%%%%%%%%%%%%%%%%%%%%%%%%%%%%%%%%%%%%%%%%%%%
\begin{abstract}
%%%%%%%%%%%%%%%%%%%%%%%%%%%%%%%%%%%%%%%%%%%%%%%%%%%%%%%%%%%%%%%%%%%%%%%%%%%%%%%%

Adaptive coding faces the following problem: given a collection of source classes such that each class in the collection has non-trivial minimax redundancy rate, can we design a single code which is asymptotically minimax over each class in the collection? In particular, adaptive coding makes sense when there is no universal code on the union of classes in the collection. In this paper, we deal with classes of sources over an infinite alphabet, that are characterized by a  dominating envelope. We provide asymptotic equivalents for the redundancy of envelope classes enjoying a regular variation property. We finally construct a computationally efficient  online prefix code, which interleaves the encoding of the so-called pattern of the message and the encoding of the dictionary of discovered symbols. This code is shown to be adaptive, within a $\log\log n$ factor, over the collection of regularly varying envelope classes. The code is both simpler and less redundant  than previously described contenders. In contrast with previous attempts, it also  covers the full range of slowly varying envelope classes. 
\end{abstract}


\section{Introduction}

%%%%%%%%%%%%%%%%%%%%%%%%%%%%%%%%%%%%%%%%%%%%%%%%%%%%%%%%%%%%%%%%%%%%%%%%%%%%%
%%%%%%%%%%%%%%%%%%%%%%%%%%%%%%%%%%%%%%%%%%%%%%%%%%%%%%%%%%%%%%%%%%%%%%%%%%%
%%\section{Setting and notation}
%%%%%%%%%%%%%%%%%%%%%%%%%%%%%%%%%%%%%%%%%%%%%%%%%%%%%%%%%%%%%%%%%%%%%%%%%%%
%%%%%%%%%%%%%%%%%%%%%%%%%%%%%%%%%%%%%%%%%%%%%%%%%%%%%%%%%%%%%%%%%%%%%%%%%%%%%%

Lossless {coding} consists of mapping in a one-to-one way finite sequences of symbols $x_{1:n}=x_1,...,x_n$ (also called messages) from  a finite or countably infinite {alphabet} $\cX$ to binary sequences so as to minimize the expected length of codewords. The mapping (or code) is not only assumed to be one-to-one but also  uniquely decodable (any concatenation of codewords can not be parsed into codewords in more than one way). Sequences are assumed to be generated by a {stationary and memoryless source}, defined as a product probability measure $P^n$, where $P\in\cM_1(\cX)$ is a probability measure on the alphabet $\cX$. 

Throughout the paper, we identify  uniquely decodable codes with probability distributions (thanks to the Kraft-McMillan inequality) and the notation  
 $Q_n\in\cM_1(\cX^n)$ is used to denote  \emph{coding distributions}. 
 Under this identification, the length of the codeword assigned to a sequence $x_{1:n}$ is at most $\lceil -\log Q_n(x_{1:n})\rceil$ (here and throughout the paper, $\log$ denotes the base-$2$ logarithm while $\ln$ denotes the natural logarithm) \citep{cover:thomas:1991}. The \emph{expected redundancy} of the coding distribution $Q_n$ (the expected number of additional bits beyond the entropy used by $Q_n$ to encode a sequence generated by the source)  then corresponds to the Kullback-Leibler divergence (or relative entropy) between $P^n$ and $Q_n$: $
D(P^n, Q_n) :=\sum_{x_{1:n}\in \cX^n} P^n(x_{1:n}) \log \frac{P^n(x_{1:n})}{Q_n(x_{1:n})} \, . $
In decision theoretic language, redundancy is also called cumulative entropy risk \citep{MR1604481}. 

When facing a source class $\cC$ with a common alphabet $\cX$ rather than a single source, the \emph{universal coding} problem consists of finding codes that 
perform well over the entire source class $\cC$. For a given class $\cC$ of  sources on $\cX$, we define $\cC^n=\left\{P^n : P\in\cC\right\}$, the class of product distributions induced by $\cC$ on $\cX^n$. 
The performance of a code (a coding distribution $Q_n$ over $\cX^n$) with respect to a source class is quantified by the \emph{maximal expected redundancy} defined as
$$
 \RED(Q_n, \cC^n)=\sup_{P\in\cC} D(P^n, Q_n) \, .
$$
The infimum of $\RED(Q_n,\cC^n)$ over all $Q_n$, is called the \emph{minimax redundancy} of $\cC^n$:
$$
 \RED(\cC^n)=\inf_{Q_n\in {\cM}_1\left(\cX^n\right)} \RED(Q_n,\cC^n)\, .
$$ 
Minimax redundancy quantifies the hardness of universal coding with respect to $\cC^n$. A source class $\cC$  is said to have a non-trivial \emph{minimax redundancy rate} if $\RED(\cC^n)= o(n)$. In the language of mathematical statistics, 
universal coding is the information theoretical counterpart of density estimation under cumulative entropy risk \citep[See][for a scholarly presentation of this correspondence]{MR1604481}.  

There are situations where universal coding is achievable and where minimax redundancy rates are precisely known. The simplest and most important setting consists of stationary memoryless sources over finite alphabets ($|\cX| =k$) where  
$$
\RED(\cC^n) = \frac{k-1}{2} \log \frac{n}{2 \pi \mathe} + O_k(1) \, ,
$$ 
as demonstrated in a series of paper that culminate with
 \citep{barron:clarke:1990,xie:barron:1997,xie:barron:2000,barron:clarke:1994}. (The notation $O_k(1) $ indicates that the upper bound may depend on $k$). Moreover it is known that in this setting optimal coding distributions
 are mixtures of product distributions. Dirichlet distributions with parameters $1/2, \ldots, 1/2$ over the $k-1$ dimensional simplex define the so-called Krichevsky-Trofimov coding distributions which form an important building block in our construction \citep[See][for details]{gassiat2014codage}.

This suggests that, while  redundancy grows only logarithmically in the message length, it grows rapidly with the alphabet size. Recent 
results provide us with a refined picture when alphabet size grows with message length \citep{SzpWei10,2015arXiv150408070F,YanBar13,SaAnKaSz14}.

If the considered  collection of sources is too large, minimax redundancy may turn out to be trivial (scaling linearly with message length). 
In other settings, the source class may be the union of smaller classes with widely differing minimax redundancy rates 
(for example sources defined by finite context trees over a finite alphabet have redundancy rates that depend on the shape of the context tree). 
 
 \emph{Adaptive coding} then considers an appropriate,  more general approach. Assume that the excessively large collection of sources is the union of smaller subclasses and  that, for each subclass, the minimax redundancy rate is non trivial and  a good universal coder is available.  Is it then possible to engineer a single coding method that performs well over all subclasses in the collection? 

Let $(\cC(\alpha))_{\alpha \in A}$ be a collection of source classes. A sequence $(Q_n)_{n\geq 1}$ of coding probabilities is said to be \emph{asymptotically adaptive} with respect to  $(\cC(\alpha))_{\alpha\in A}$ if for all $\alpha\in A$
\begin{equation} \label{eq:adaptive}
\RED(Q_n, \cC(\alpha)^n)= \sup_{P\in\cC(\alpha)} D(P^n,Q_n) \leq (1+o_\alpha(1))\RED(\cC(\alpha)^n)
\end{equation}
as $n$ tends to infinity. If the inequality \eqref{eq:adaptive} holds with a factor other than $(1+o_\alpha(1))$ (that may depend on the source class $\cC(\alpha)$) larger than $1$ to the right, then we say that there is adaptivity \emph{within} this factor. Context-Tree-Weighting \citep{Wil98} provides an example of an adaptive code with respect to sources with bounded or unbounded memory over finite alphabets \citep{catoni:2004,Gar06}. 

 The notion of adaptivity comes from mathematical statistics \citep{MR2013911}, and adaptive coding is much related to competitive estimation \citep{Sur16}. It is also known as twice-universal coding \citep{Rya1984} or hierarchical universal coding \citep{MeFe98}.
%We conform to the conventions and definitions of mathematical statistics  \citep{MR2013911}.   

When facing countably infinite alphabets, adaptive coding is  a natural problem: when the alphabet is infinite, a theorem due to Kieffer \citep{MR514346} entails that there does not exist a code $Q_n$ such that for all distribution $P$ over $\cX$, $\frac{D(P^n,Q_n)}{n}\to 0$, that is, there is no universal code for the entire class of stationary memoryless sources over $\cX$ and for this class the minimax redundancy rate is trivial. 

%%%%%%%%%%%%%%%%%%%%%%%%%%%%%%%%%%%%%%%%%%%%%%%%%%%%%%%%%%%%%%%%%%%%%%%%%%%%%%%%%%%%%%%%%%%%%%%%%%%%%%%%%%%%%%%%%%%%%%%%%%%%%%%%%%%%%%%%%%%%%%%%%%%%%%%%%%%%%%%
%%%%%%%%%%%%%%%%%%%%%%%%%%%%%%%%%%%%%%%%%%%%%%%%%%%%%%%%%%%%%%%%%%%%%%%%%%%%%%%%%%%%%%%%%%%%%%%%%%%%%%%%%%%%%%%%%%%%%%%%%%%%%%%%%%%%%%%%%%%%%%%%%%%%%%%%%%%%%%%
%%%%%%%%%%%%%%%%%%%%%%%%%%%%%%%%%%%%%%%%%      the case for infinite alphabets                          %%%%%%%%%%%%%%%%%%%%%%%%%%%%%%%%%%%%%%%%%%%%%%%%%%%%%%%
%%%%%%%%%%%%%%%%%%%%%%%%%%%%%%%%%%%%%%%%%      build from    MR2097043,MR2095850                        %%%%%%%%%%%%%%%%%%%%%%%%%%%%%%%%%%%%%%%%%%%%%%%%%%%%%%%
%%%%%%%%%%%%%%%%%%%%%%%%%%%%%%%%%%%%%%%%%%%%%%%%%%%%%%%%%%%%%%%%%%%%%%%%%%%%%%%%%%%%%%%%%%%%%%%%%%%%%%%%%%%%%%%%%%%%%%%%%%%%%%%%%%%%%%%%%%%%%%%%%%%%%%%%%%%%%%%
%%%%%%%%%%%%%%%%%%%%%%%%%%%%%%%%%%%%%%%%%%%%%%%%%%%%%%%%%%%%%%%%%%%%%%%%%%%%%%%%%%%%%%%%%%%%%%%%%%%%%%%%%%%%%%%%%%%%%%%%%%%%%%%%%%%%%%%%%%%%%%%%%%%%%%%%%%%%%%%

This   result has not deterred investigations on coding against countable alphabets. 
There are indeed sensible motivations for investigating such a framework: \citeauthor{MR2097043} \citep{MR2097043} observe that applications of compression techniques  (text, image and so on) usually involve alphabets that are very large (words or even groups of words) and that  the traditional transformation of words into letters and then into bits hides useful  dependencies.  Kieffer's theorem  prompted at least two different approaches: 
 \begin{enumerate}[i)]
 \item \citet{MR2097043} separate the description of
strings over large alphabets into two parts: description of the
symbols appearing in the string, and of their pattern, the order
in which the symbols appear. They redefine the performance criterion by focusing on compressing the message's \emph{pattern}  \citep{MR2097043,MR2095850,garivier:2006,2015arXiv150408070F,shamir-2006};
\item investigating the redundancy on smaller source classes that satisfy Kieffer's condition. The so-called envelope classes investigated in \citep{boucheron:garivier:gassiat:2006} form an example of such classes \citep[See also][]{acharya2014poissonization,bontemps2012adaptiveb,boucheron2015adaptive}.  
  \end{enumerate} 

This paper pursues both lines of research: we deal with the collection of so-called envelope classes, but the adaptive code we introduce and investigate will turn out to be a pattern encoder. In contrast with \citep{MR2097043,MR2095850}, we handle simultaneously  dictionary and pattern encoding. 

\begin{dfn}[\textsc{envelope classes and envelope distributions}] \label{dfn:envelope:class}
Let $f$ be a non-increasing mapping from $\Np$ to $(0,1],$ with ${1< \sum_{j\in\Np} f(j)<\infty}$. The \emph{envelope class} $\cC(f)$ defined by the function $f$ is the collection of distributions which are dominated by $f$:
$
\cC(f)
 :=\big\{ P : \, \forall j\in \Np,\;p_j\leq
f(j)~\big\}\, .$ Define
$ \ell_f=\min\big\{\ell\geq 1,\, \sum_{j=\ell}^{+\infty}f(j)\leq 1\big\}. $
The associated \emph{envelope distribution} $F$ is defined as
$ F(k) :=  1 - \sum_{j>k} f(j)$ if $k+1\geq\ell_f$, and $F(k):=0$ otherwise. The \emph{envelope probabilities} $(f_j)_{j\geq 1}$ are  defined by $f_j=F(j)-F(j-1)$.
\end{dfn}

Envelope classes provide a framework where the search for adaptive coding strategies is feasible \citet{2015arXiv150408070F}.

%{\textcolor{red}{Le bout qui suit jusque "is desired" me parait bizarre ici, comme si on disait c'est debile de faire des enveloppes (donc avec ordre) mais on le fait car on ne peut rien faire d'autre. J'aurais presque envie de retirer ces 5 lignes; ou bien on pourrait juste mettre la ref 20 en fin de la phrase d'avant sans commentaire -- c'est la ref qui dit que oui les enveloppes c'est le cadre qui permet de coder adaptatif en infini. Non ?}}Their relevance to practical applications may be questioned:  point out that in natural language processing where the (large but finite) alphabet may be considered as the set of words of some natural language, there is no special reason  to privilege any ordering on the alphabet whereas this 
%is implicitly done when defining envelope classes. But the tight redundancy bounds they establish  
%reveal that when dealing with infinite alphabets, some ordering has to be assumed if finite redundancy is desired. 

The contribution of our paper is two-fold:
\begin{enumerate}[i)]
\item we use Poissonization arguments introduced by \citet{acharya2014poissonization} to derive tight bounds on the minimax redundancy of envelope classes. Karlin's framework (see Section \ref{subsec:karlin-forma}) then provides a tractable interpretation of those bounds, and allows us to readily obtain asymptotic equivalents when the envelope satisfies a so-called regular variation property. We thus fill the gap between previously proved upper and lower bounds for minimax redundancies of such classes;
\item we construct a simple coding scheme, which we call the Pattern Censoring Code (\textsc{pc} code). This code performs an online encoding or decoding of a sequence of symbols. It pertains to the family of censoring codes   described in \citep{boucheron:garivier:gassiat:2006,Bon11,bontemps2012adaptiveb,boucheron2015adaptive}:  in the \textsc{pc} code, first occurrences of symbols are censored, that is they are encoded using a general purpose encoder for integers (such as Elias encoding \citep{MR0373753}) and are implicitly inserted into a dictionary; symbols that have already been inserted into the dictionary are translated into their rank of insertion in the dictionary and then fed to a Krichevsky-Trofimov encoder that works on an alphabet that matches the current size of the dictionary. The Krichevsky-Trofimov encoder actually performs a variant of pattern coding as introduced in \citep{MR2097043}. We show that the \textsc{pc} code is adaptive, within a $\log\log n$ factor, over the collection of envelope classes whose decay is regularly varying. We thus get, with a single code, almost adaptivity over envelope classes for the whole range of possible regular variation.
\end{enumerate}

The paper is organized as follows. In Section \ref{sec:notat}, we lay down the notation and introduce Karlin's framework and regular variation. The main results are stated in Section \ref{sec:main-results}. In Section \ref{sec-minmaxred}, we 
provide a  sketch of the derivation of the redundancy bounds by combining the Poissonization techniques of \citet{acharya2014poissonization} and Karlin's framework. In particular, when the envelope $f$ has a so-called \emph{light tail} (a regular variation index equal to zero), we show that $\RED(\cC(f)^n)$ is asymptotically equivalent to 
\begin{eqnarray*}
\bR_f(n)&:=&\log(\mathe)\int_1^n\frac{\vec\nu_f(1/t)}{2t}\mathd t\, ,
\end{eqnarray*} 
where $\vec\nu_f(x)=|\{j\geq 1,\, f_j\geq x\}|$.
This characterization is a powerful generalization of the tight bounds that have been established for memoryless sources over finite alphabets. The latter can be regarded as envelope classes where $\vec\nu_f(x)=k$ for some $k$  and all small enough $x$. Indeed $\frac{k-1}{2}\log n$ scales like $\bR_f(n)$ with respect to both $k$  and $n$. The quantity $\bR_f(n)$ also provides an equivalent of the minimax redundancy for envelope classes defined by log-concave envelopes (such that $f(k)f(k+2)/f(k+1)^2\leq 1$ for all $k\geq 1$) that was characterized in \citep{bontemps2012adaptiveb}. Up to a constant factor, the minimax redundancy for envelope
classes defined by heavy-tailed envelopes (with positive regular variation indexes, as investigated in \cite{boucheron2015adaptive}) also scales like $\bR_f(n)$. This is the content of Theorem \ref{thm:red-asymp}. 

The bounds on minimax redundancy rates in Karlin's framework also suggest a universal coding strategy for each envelope class. In words, when encoding the $n^{\text{th}}$ symbol in the message, it seems sensible to handle symbols with probability larger than $1/n$ (frequent symbols) differently from symbols with probability  smaller than $1/n$ (rare symbols). The probability of frequent symbols can be  faithfully estimated while the probability of rare symbols can barely be estimated from the message. This is widely acknowledged in the blossoming literature concerning estimation of discrete probability distributions, see  \citep{valiant2011estimating}. The censoring code approach described in \citep{boucheron:garivier:gassiat:2006,Bon11,bontemps2012adaptiveb,boucheron2015adaptive} explores that kind of path. Here, we combine pattern coding and censoring so as to manufacture a simple encoder that achieves adaptivity 
within a $\log \log n$  factor with respect to all regularly varying envelope classes (except the extreme case of \emph{very heavy-tailed} envelopes, corresponding to an index equal to $1$, see Section \ref{subsec:reg-var-def}). The Pattern Censoring Code is described in Section \ref{sec:bhc-code}, and Section \ref{sec:code_analysis} provides a modular analysis of its redundancy. Proofs are gathered in Section \ref{sec:proofs}.

%%%%%%%%%%%%%%%%%%%%%%%%%%%%%%%%%%%%%%%%%%%%%%%%%%%%%%%%%%%%%%%%%%%%%%%%%%%%%%
%%%%%%%%%%%%%%%%%%%%%%%%%%%%%%%%%%%%%%%%%%%%%%%%%%%%%%%%%%%%%%%%%%%%%%%%%%%%%%%%%
%\input{background}
%%%%%%%%%%%%%%%%%%%%%%%%%%%%%%%%%%%%%%%%%%%%%%%%%%%%%%%%%%%%%%%%%%%%%%%%%%%%%%
%!TEX root = envelope.tex
%%%%%%%%%%%%%%%%%%%%%%%%%%%%%%%%%%%%%%%%%%%%%%%%%%%%%%%%%%%%%%%%%%%%%%%%%%%

\section{Background and notation} % (fold)
\label{sec:notat}

% section notat (end)
From now on, the message alphabet $\cX$ is the set of positive integers $\N_+ := \N \setminus \{0\}$, and the source probabilities will be denoted by their probability mass functions $(p_j)_{j\geq 1}$. 

\subsection{Occupancy counts}
\label{subsec:occup-counts}

Define $N_n^j$ as the empirical count of symbol $j$, that is the number of times symbol $j$ occurs in a sample of size $n$. The sequence $(N_n^j)_{j\geq 1}$ is called the \emph{type} of the sample. Let $K_{n,r}$ denote the \emph{occupancy counts}, defined as the number of symbols that appear exactly $r$ times in a sample of size $n$:
$$
  K_{n,r}= \sum_{j \geq 1} \IND_{\{N_n^j=r\}}\,.
$$
The sequence $(K_{n,r})_{r\geq 1}$ is called the \emph{profile} of the sample. The occupancy counts combine to yield the total number of distinct symbols in the sample, denoted by $K_n$:
$$
  K_n=\sum_{j\geq 1} \IND_{\{N_n^j>0\}} = \sum_{r \geq 1} K_{n,r} \, .
$$
We will also encounter a quantity known as the \emph{missing mass}, which corresponds to the cumulated probability of the unseen symbols:
$$
M_{n,0}=\sum_{j \geq 1} p_j \IND_{\{N_n^j=0\}}\, .
$$
Those quantities play a central role in estimating various properties of the sampling distribution, especially in situations where the sample size is small compared to the alphabet size \citep{orlitsky2004modeling,valiant2011estimating,jiao2015minimax}. In light of the correspondence  between compression and probability estimation, they also are crucial to the analysis of the coding strategy we investigate here.

Henceforth,  $\EXP$ will  denote  expectation with respect to the source distribution $P$ or its $n$-fold product $P^n$, while $\EXP_f$ will denote the expectation with respect to the envelope distribution $(f_j)_{j\geq 1}$ or its $n$-fold product.

\subsection{Karlin's framework}
\label{subsec:karlin-forma}

In \citeyear{MR0216548}, S. Karlin introduced an elegant framework for investigating the asymptotic properties of infinite urn models. 
Sources over countable alphabets fit nicely into this framework. 
This framework has recently received attention in random combinatorics and stochastic processes theory \citep[See][for a survey]{GneHanPit07}. 	
Although the  framework had been  introduced  in order to perform a thorough exploration of asymptotic normality
of occupancy counts, it  has also 
proved convenient when assessing the tightness of non-asymptotic tail or moment bounds for occupancy counts as demonstrated in \citep{ben2014concentration}. 

A probability mass function $(p_j)_{j \geq 1}$ defines a  \emph{counting measure} $\nu$ defined by
$
\nu(\mathd x)=\sum_{j\geq 1} \delta_{p_j}(\mathd x)\, ,
$
where $\delta_p$ is the Dirac mass at $p$. Let the counting function $\vec\nu(\cdot)$ be the right tail of $\nu$, that is for all $x\in(0,1]$,
$$
\vec\nu(x)=\nu[x,\infty[= |\{j\geq 1,\, p_j\geq x\}|\, .
$$

Let us also define the measure $\nu_1$ by
$$
\nu_1(\mathd x)=\sum_{j\geq 1}p_j\delta_{p_j}(\mathd x)\, .
$$
For $x\in[0,1]$, $\nu_1[0,x]=\sum_{j\geq 1}p_j\IND_{p_j\leq x}$ is the cumulated probability of symbols with probability smaller than $x$. Expected occupancy counts and masses can be written simply as integrals against the measure $\nu$. For instance
\begin{eqnarray*}
\EXP K_n&=& \int_0^1 \left(1-(1-x)^n\right)\nu(\mathrm{d} x)\, .
\end{eqnarray*}
The sequence $(\EXP K_n)_n$ characterizes the measure $\nu$ (this follows easily from a Laplace transform argument).

\subsection{Regular variation and regularly varying envelopes}
\label{subsec:reg-var-def}

The envelope distribution considered in this paper will most often satisfy a so-called \emph{regular variation} property. A  fascinating treatment of regular variation can be found in the textbook \citep{BinGolTeu89}. 

\begin{dfn} A measurable function $g : ]a, \infty) \rightarrow ]0,\infty)$ (where $a>0$) is said to be regularly varying at infinity if for all $x>0$ \[
	\lim_{t \to \infty} \frac{g(tx)}{g(t)} = x^\alpha
\]  	
	for some $\alpha \in \mathbb{R}$. This is summarized by $g\in \textsc{rv}_ \alpha$.  If $\alpha=0$, $g$ is said to be slowly varying. Note that $g\in\RV_\alpha$ if and only if, for all $t>a$, $g(t)=t^\alpha\ell(t)$ where $\ell$ is slowly varying.
\end{dfn}

Useful results from regular variation theory are gathered in Appendix \ref{sec:reg-var-tools}.  

Following \citet{MR0216548} and \citet*{GneHanPit07}, source classes with regularly varying envelopes are defined as follows.

\begin{dfn}[\textsc{regularly varying source classes}]
The envelope class $\cC(f)$ is said to be regularly varying with index $\alpha\in[0,1]$ if the function $\vec\nu_f(1/\cdot)$ is regularly varying with index $\alpha$ (denoted $\vec\nu_f(1/\cdot)\in\RV_\alpha$).
\end{dfn}
Note that the property pertains to the envelope counting function. The counting functions associated with source distributions in the class may or may not satisfy the regular variation property.

%% Some help 
If the envelope distribution is a so-called power-law distribution, that is if, for some  $C>0$ and $0 < \alpha <1$, $f(j) \sim  C j^{-1/\alpha}, j \to \infty,$ then $\vec\nu_f$ is  regularly varying with index $\alpha$ \citep[See][for details]{acharya2014poissonization}. 

The case $\alpha=0$ corresponds to light-tailed envelopes. It contains for instance frequencies proportional to $C \exp\left(-C' j^ \beta\right)$, for some $\beta>0$ (in particular, when $\beta=1$, the Geometric distribution), or for example frequencies corresponding to the integral part of a log-normal variable as treated in \citep{bontemps2012adaptiveb}.

Under the regular variation assumption, the asymptotic expected values of $K_n$, $K_{n,r}$ and $M_{n,0}$ are nicely related to the counting function $\vec\nu$ (see Appendix \ref{sec:infinite_urn_sch}). Appendix \ref{sec:connections_between_countin} relates properties of the counting function and properties of the quantile function of the sampling distribution.

If $f,g$ denote two functions from $\mathbb{R}_+$ to $\mathbb{R}_+$, notation $f \asymp g$ means that there exists some  constant $C\geq 1$ and $t_0 \in \mathbb{R}_+$ such that for $t>t_0$, 
$1/C g(t) \leq f(t) \leq C g(t)$.  

%%%%%%%%%%%%%%%%%%%%%%%%%%%%%%%%%%%%%%%%%%%%%%%%%%%%%%%%%%%%%%%%%%%%%%%%%%%%%%
%%%%%%%%%%%%%%%%%%%%%%%%%%%%%%%%%%%%%%%%%%%%%%%%%%%%%%%%%%%%%%%%%%%%%%%%%%%%%%%%%
%\input{main-results}
%%%%%%%%%%%%%%%%%%%%%%%%%%%%%%%%%%%%%%%%%%%%%%%%%%%%%%%%%%%%%%%%%%%%%%%%%%%%%%
%!TEX root = envelope.tex

\section{Main results} \label{sec:main-results}

% section main_results (end)

We first describe the minimax redundancy of envelope classes $\cC(f)^n$ in  Karlin's framework, that is using the counting function $\vec \nu$ and measure $\nu_1$. 
The first theorem follows from the Poissonization approach promoted in \citep{YanBar13} and 
\citep{acharya2014poissonization} \citep[See also][]{jacquet1998analytical}.  It  gathers upper and lower bounds on the minimax redundancy of envelope classes. Note that the bounds do not depend on any regular
variation assumption on the envelope function.  
\begin{thm}\label{prop:red-upper}
For any envelope function $f$, the minimax redundancy of $\cC(f)^n$ satisfies
\begin{enumerate}[i)]
\item 
\begin{eqnarray}\label{eq:red-upper}
\RED\left(\cC(f)^n\right)&\leq& \log(\mathe)\left(\int^{n}_1\frac{\vec\nu_f(1/t)}{2t}\mathd t + \vec\nu_f(1/n)+ n\nu_{1,f}[0,1/n]\right)+O(\ell_f\log n)\, .
\end{eqnarray}
\item  Let $m=n-n^{2/3}$. For any envelope function $f$, there exists a constant $c_f>0$ such, for large enough $n$,
\begin{eqnarray}\label{eq:red-lower}
\RED\left(\cC(f)^{n}\right)&\geq &\left( \log(\mathe)\int_1^m \frac{\vec\nu_f(1/t)}{2t}\mathd t -5\vec\nu_f(1/m)-1\right)\vee \left(\EXP_f[K_n]-c_f\right)\, .
\end{eqnarray}

\end{enumerate}\end{thm}

Note that, as soon as the support of $(f_j)_{j\geq 1}$ is infinite, $\RED\left(\cC(f)^n\right)\gg\log n$. Hence, in our context, the term $O(\ell_f\log n)$ in \eqref{eq:red-upper} is asymptotically negligible with respect to the other terms.

For an envelope function $f$ defined by frequencies $(f_j)_{j\geq 1}$, with counting function $\vec \nu_f$, recall that $\bR_f(n)$ is defined as
\begin{eqnarray*}
\bR_f(n)&=&\log(\mathe)\int_1^n\frac{\vec\nu_f(1/t)}{2t}\mathd t\, .
\end{eqnarray*} 
Regular variation assumptions on the envelope allow us to compare the different terms  in \eqref{eq:red-upper} and \eqref{eq:red-lower} (see Theorems \ref{thm:karamata} and \ref{thm:asympt-occup} in Appendix \ref{sec:reg-var-tools} and \ref{sec:infinite_urn_sch}). 
Assuming that the counting function $\vec \nu_f$ satisfies a regular variation property, the sequence $\left(\bR_f(n)\right)_n$ is proportional to the redundancy rate of the envelope class $\cC(f)$. It will also turn out to characterize the maximum redundancy of the \textsc{pc} code over $\cC(f)$.

 When $\vec\nu_f(1/\cdot)\in\RV_0$,  Karamata's integration Theorem and Proposition 15 in \citep{GneHanPit07} imply \[
 	\bR_f(n) \gg \vec\nu_f(1/n) \underset{n \to \infty}\sim \EXP_f[K_n] \gg n\nu_{1,f}[0,1/n] \, . 
 \]
Hence if  $\vec\nu_f(1/\cdot)\in\RV_0$, Theorem \ref{prop:red-upper} entails $\RED(\cC(f)^n) \underset{n \to \infty}\sim \bR_f(n)$.

When $\vec\nu_f(1/\cdot)\in\RV_1$, the largest  term in the upper bound is now $n\vec\nu_{1,f}[0,1/n]$, which is of order $\EXP_f [K_n]\gg \vec\nu_f(1/n)$.

Last but not least, when $\vec\nu_f(1/\cdot)\in\RV_\alpha$ for $0<\alpha<1$, all three  terms in the right-hand-side of \eqref{eq:red-upper}  have the same order of magnitude $\EXP_f[K_n]\asymp \vec\nu_f(1/n)\asymp\bR_f(n)$. Hence,  the following theorem  
follows directly from Theorem \ref{prop:red-upper}.

\begin{thm}[\cite{acharya2014poissonization}]
\label{thm:red-asymp}
Let $\cC(f)$ be an envelope source class, with $f\in\RV_\alpha$ and $\alpha\in[0,1[$. Then
\begin{eqnarray*}
\RED(\cC(f)^n)&\asymp & \bR_f(n)\, .
\end{eqnarray*}
If $\alpha =0$, \[
\lim_{n \to \infty}	\frac{\bR_f(n)}{\RED(\cC(f)^n)}  =1  \, .
\]
If $\alpha=1$, 
\begin{eqnarray*}
\RED(\cC(f)^n)&\asymp & \EXP_f[K_n]\gg \bR_f(n)\, .
\end{eqnarray*}
\end{thm}

%% Interpretation of $\bR_f(n)$ : coding symbols with probability larger than $1/n$. 

% Discusssion 
% Compare with bounds in Bontemps et al.
%% Maybe not here.  
The implementation of the Poissonization technique by \citeauthor{acharya2014poissonization}  answered a number of open questions raised in \citep{boucheron:garivier:gassiat:2006,bontemps2012adaptiveb}. Previously, the minimax redundancy rate for source classes defined by regularly varying envelopes was only known when  envelopes were discrete log-concave (which implies non-decreasing hazard rate), that is for a strict subset of classes defined by slowly varying envelopes. Before \citep{acharya2014poissonization}, the best known general upper bound on the redundancy of envelopes classes was stated using the tail envelope function and read as 
\begin{eqnarray*}\label{eq:upper-bound-2009}
\RED(\cC(f)^n)&\leq & \inf_{u\leq n} \left\{ n\overline{F}(u)\log\mathe +\frac{u-1}{2}\log n\right\} +2\, 
\end{eqnarray*} where $\overline{F}(u)=\sum_{j>u}f_j$  \citep{boucheron:garivier:gassiat:2006}. 

Optimizing over  $u$  leads to outline $u=\vec{\nu}_f(\tfrac{\ln n}{2n})$. The redundancy upper bound then translates into 
\begin{equation}
\label{eq:old:red:upper:bound}
 	\RED(\cC(f)^n) \leq \left(\frac{2n}{\ln n} \nu_{1,f} \left( 0, \tfrac{\ln n}{2 n}		\right) + \vec{\nu}_f\left(\frac{\ln n}{2 n}\right)\right) \frac{\log n}{2} \, . 
\end{equation} 

% %
When $\vec \nu_f(1/t) \underset{t \to \infty}\sim t^ \alpha \ell(t)$ with $\ell \in \RV_{0}, \alpha \in (0,1)$, combining \eqref{eq:old:red:upper:bound} and Proposition 13 from \citep{GneHanPit07} reveals that the upper bound in \eqref{eq:old:red:upper:bound} is asymptotically equivalent to \[
	\frac{\log \mathe }{2^{1- \alpha} (1 -	\alpha)} 
	   \ell \left(\tfrac{2 n}{\ln n}\right)   n^ \alpha  ({\ln n})^{1- \alpha}\, . 
\]
Invoking Proposition 13 from \citep{GneHanPit07} and Karamata's integration Theorem show that the right-hand-side of \eqref{eq:red-upper} is asymptotically equivalent to \[
 \frac{\log(\mathe)(1+ \alpha)}{2 \alpha(1- \alpha)} \ell(n) n^\alpha
    \, . 	
\]
For envelopes with positive regular variation index $\alpha$, the redundancy upper bound from \citep{boucheron:garivier:gassiat:2006} is off by a factor $(\ln n)^{1- \alpha}$. 

For  envelopes that satisfy the discrete log-concavity condition (a special case of envelope function with non-decreasing hazard rate as treated in \citep{bontemps2012adaptiveb}), then by Lemma \ref{lem:nu2U}, the redundancy bounds derived in \citep{bontemps2012adaptiveb} and    
$\bR_f(n)$ are asymptotically equivalent. 

An interesting benefit from the Poissonization method concerns slowly varying envelopes that do not satisfy the discrete log-concavity condition. For slowly varying envelopes, the upper bound in \eqref{eq:upper-bound-2009} is asymptotically equivalent to $ \vec{\nu}_f\left(\frac{\ln n}{2 n}\right) \frac{\log n}{2}.$  If we focus on $\vec\nu_f(1/t)= \lfloor \exp\left( (\ln t)^\beta\right)\rfloor $ for $\beta \in (0,1), t>1$, then, integration by part shows that 
 $ \bR_f(n) \leq \vec \nu_f(1/n) (\ln (n))^{1- \beta}/ \beta.$ Under this setting, for $\beta \in (0,1/2)$, by Theorem  from \citep{BinGolTeu89}, $\vec{\nu}_f\left(\frac{\ln n}{2 n}\right)\underset{n \to \infty}\sim \vec \nu_f(1/n)$. This is enough to conclude that the redundancy upper bound derived from \eqref{eq:upper-bound-2009} and the bound derived from \eqref{eq:red-upper} are not of the same order of magnitude:  $\bR_f(n) =  o\big(  \vec{\nu}_f\left(\frac{\ln n}{2 n}\right) \frac{\log n}{2}\big).$ 

We now turn to adaptivity and  characterize the performance of the \textsc{pc} Code $(Q_n)_n$ defined in Section \ref{sec:bhc-code} on regularly varying envelope classes $\cC(f)$, with index $\alpha\in[0,1[$. 

A regular variation assumption on the envelope distribution is enough to characterize the redundancy rate of the source class (Theorem \ref{thm:red-asymp}). Underneath this rewarding message lies a more ambivalent finding: the redundancy rate breaks down into three components whose relative importance depends on the unknown regular variation index (Theorem \ref{prop:red-upper}). 

Experience with different attempts at designing adaptive codes  suggests that an adaptive code should be able to sort rare symbols from frequent symbols (with probability mass larger than the reciprocal of the current message length) and to work with an effective alphabet of size not larger than $\vec{\nu}_f(1/n)$. The \textsc{pc} code attempts to do this in a simple and efficient manner:    
any recurring symbol is deemed to be frequent; any symbol occurring for the first time is deemed to be rare. In this way, the size of the alphabet used by the mixture coder  coincides with the number of distinct symbols in the portion of the message encoded  so far. The latter quantity is well concentrated around its mean value and is bounded from above by the mean value under the envelope distribution. No special attempt is made to encode  the dictionary in a way that adapts to a regular variation index since the source distribution may not  satisfy any regular variation assumption. 

Despite its simplicity the \textsc{pc} code achieves adaptivity within a $\log \log n $ factor over all regularly varying envelope classes.

\begin{thm}\label{thm:main-compression}
Let $(Q_n)_n$ be the sequence of coding distributions associated with the Pattern Censoring code. For all $\alpha\in[0,1[$, for every envelope function $f$ with $\vec\nu_f(1/\cdot)\in\RV_\alpha$, there exists constants $a_f,b_f>0$ such that
\begin{equation*}%\label{eq:main-compression}
(a_f+o_f(1)) \leq \frac{\RED(\cC(f)^n)}{\bR_f(n)}\leq \frac{\RED(Q_n,\cC(f)^n)}{\bR_f(n)}\leq (b_f+o_f(1))\log\log n\, .
\end{equation*}
In particular, the Pattern Censoring code is adaptive, within a $\log\log n$ factor, with respect to the collection 
$$\big\{ \cC(f) : {f\in\RV_\alpha}, {\alpha\in[0,1[} \big\} \, . $$
\end{thm}

 The \textsc{pc} code is designed in a way that parallels the design of the \textsc{ac} and \textsc{etac} codes from \citep{bontemps2012adaptiveb,boucheron2015adaptive}. All three codes interleave mixture coding of symbols that are deemed frequent and Elias encoding of symbols that are deemed rare. They differ in the way censoring is performed. The \textsc{ac} code censors records that is symbols that are larger than symbols seen so far. The \textsc{etac} code censors the $n^{\text{th}}$ symbol if it is larger than $M_n$ where $M_n$ is a carefully chosen  order statistics of the sub-sample $X_1, X_2, \ldots, X_n$ (See Equation (\ref{def:mn}) in Appendix \ref{sec:connections_between_countin}).  At each instant, the \textsc{ac} and the \textsc{etac} codes handle an effective alphabet formed by a sequence of consecutive symbols. The \textsc{ac} code is doomed to failure if the source has a positive regular variation index $\alpha \in (0,1)$: the increasing alphabet used by the mixture  code typically grows like $n^{\alpha/(1- \alpha)}$ instead of $n^\alpha$. The increasing alphabet used by the mixture code of the \textsc{etac} code grows much more reasonably: if the source distribution  is dominated by a regularly varying envelope   $f$ with index $\alpha \in  (0,1)$, the increasing alphabet typically grows like $\EXP M_n$  which is or order $\vec \nu_f(1/n)$ (See Appendix \ref{sec:connections_between_countin}). If the envelope is slowly varying and discrete log-concave, $\EXP M_n$ also grows at a pace that is consistent with the slow variation property of the envelope. 
 But our current understanding of the rate of slow variation problem \cite[See][Section 2.3]{BinGolTeu89} does not allow us to quantify precisely the rate of growth the alphabet when the envelope is slowly varying and is equivalent to a function in the de Haan class (See Definition \ref{def:dehaan} in Appendix \ref{sec:reg-var-tools}) where the auxiliary function tends to infinity. This lack of understanding also prevents us from quantifying the redundancy of the \textsc{etac} code with respect to corresponding envelope classes.  

The \textsc{etac} and \textsc{pc} codes also differ in the way they encode the escape symbol $0$ (see below for details). The mixture encoder used by \textsc{etac} code always considers $0$  as a new symbol:  the predictive probability associated with $0$ at instant $n+1$ is $1/(n+(M_n+1)/2)$. The predictive probability assigned to $0$  should rather be close to an estimator of the probability of discovering a new symbol, that is of the missing mass. For a regularly varying sampling distribution with index $\alpha \in (0,1)$, the predictive probability assigned to $0$ should scale like $n^{\alpha-1}$ rather than $1/n$. The \textsc{pc} code assigns a predictive probability $(K_n+1/2)/(n+K_n+1)$ to the escape symbol. Even though this does not coincide with the Good-Turing estimator of the missing mass $K_{n,1}/n$, it scales correctly with $n$ for regularly varying sampling distributions. 

Theorems \ref{prop:red-upper} and  \ref{thm:red-asymp} 
are established in Section \ref{sec-minmaxred}.
Theorem \ref{thm:main-compression} is proved  in Section \ref{sec:code_analysis}. 
Useful technical results and arguments are stated in Section \ref{sec:proofs}.

%%%%%%%%%%%%%%%%%%%%%%%%%%%%%%%%%%%%%%%%%%%%%%%%%%%%%%%%%%%%%%%%%%%%%%%%%%%%%
%%%%%%%%%%%%%%%%%%%%%%%%%%%%%%%%%%%%%%%%%%%%%%%%%%%%%%%%%%%%%%%%%%%%%%%%%%%%%%
%\input{minimax-red-bounds}
%%%%%%%%%%%%%%%%%%%%%%%%%%%%%%%%%%%%%%%%%%%%%%%%%%%%%%%%%%%%%%%%%%%%%%%%%%%%%%
%!TEX root = envelope.tex

%%%%%%%%%%%%%%%%%%%%%%%%%%%%%%%%%%%%%%%%%%%%%%%%%%%%%%%%%%%%%%%%%%%%%%%%%%%
%%%%%%%%%%%%%%%%%%%%%%%%%%%%%%%%%%%%%%%%%%%%%%%%%%%%%%%%%%%%%%%%%%%%%%%%%%%
\section{Minimax redundancy of envelope source classes}
\label{sec-minmaxred}
%%%%%%%%%%%%%%%%%%%%%%%%%%%%%%%%%%%%%%%%%%%%%%%%%%%%%%%%%%%%%%%%%%%%%%%%%%%
%%%%%%%%%%%%%%%%%%%%%%%%%%%%%%%%%%%%%%%%%%%%%%%%%%%%%%%%%%%%%%%%%%%%%%%%%%%

This section describes upper and lower bounds for the minimax redundancy of envelope source classes. The techniques are  borrowed from \citep{acharya2014poissonization} where tight bounds on \emph{minimax regret} (see below) are derived thanks to  Poissonization arguments. Not too surprisingly, Poissonization allows to derive tight bounds for minimax redundancy as well. 
In Karlin's framework these bounds are illuminating. 
When applied to regularly varying envelopes with positive regular variation parameter,  the resulting bounds are tight up to constant factors. When applied  to slowly varying envelopes, the resulting bounds actually imply asymptotic equivalence.

We now give the main ingredients leading to Theorem \ref{prop:red-upper}. The upper bound \eqref{eq:red-upper} follows from minimax regret bounds in \cite{acharya2014poissonization}. Upper bounds are then translated  using Karlin's framework, this  provides us with insights onto effective coding methods. In the other direction, we need to show that Poissonization arguments also apply to minimax redundancy, and to establish lower bounds on the redundancy of Poisson classes. 

%%%%%%%%%%%%%%%%%%%%%%%%%%%%%%%%%%%%%%%%%%%%%%%%%%%%%%%%%%%%%%%%%%%%%%%%%%%
\subsection{Properties of minimax redundancy}
%%%%%%%%%%%%%%%%%%%%%%%%%%%%%%%%%%%%%%%%%%%%%%%%%%%%%%%%%%%%%%%%%%%%%%%%%%%%

The next theorem collects  properties of minimax redundancies \citep[See][for a general perspective]{gassiat2014codage}.

\begin{prop}\label{prop:red-properties} The sequence $(\RED(\cC^n))_{n\geq 1}$ of minimax redundancies of a class $\cC$ of stationary
memoryless sources satisfies 
\begin{enumerate}[i)]
\item $(\RED(\cC^n))_{n\geq 1}$ is non-decreasing.
\item $(\RED(\cC^n))_{n\geq 1}$ is sub-additive: 
\begin{eqnarray*}
\text{for all} \quad  n,m\geq 1, \quad \RED(\cC^{n+m})&\leq & \RED(\cC^n)+\RED(\cC^m)\, .
\end{eqnarray*}
\item The minimax redundancy is equal to the maximin  Bayesian redundancy:
\begin{eqnarray*}
\RED(\cC^n)&=& \sup_{\pi\in\cM_1(\cC)}\inf_{Q_n\in\cM_1\left(\cX^n\right)}\int D(P^n,Q_n)\mathd\pi(P)\\
&=& \sup_{\pi\in\cM_1(\cC)} \int D(P^n, P_\pi^n)\mathd \pi(P)\, ,
\end{eqnarray*}
where $P_\pi^n$ is the mixture distribution given by $P_\pi(j)=\int P(j)\mathd\pi(P)$.
\end{enumerate}
\end{prop}

We will also use the notion of regret that originates from individual sequence analysis \citep{cesa-bianchi:lugosi:2006}. The regret of a coding probability with respect to a source class $\cC$ is \[
	\REG(Q_n,\cC^n) := \sup_{x_{1:n} \in \cX^n}  \sup_{P \in \cC}\log \frac{P^n(x_{1:n})}{Q_n(x_{1:n})} 
\]  
and the minimax regret of  $\cC^n$  is $\REG(\cC^n) :=  \inf_{Q_n} \REG(Q_n, \cC^n)$. 

%%%%%%%%%%%%%%%%%%%%%%%%%%%%%%%%%%%%%%%%%%%%%%%%%%%%%%%%%%%%%%%%%%%%%%%%%%%
%\subsection{Poisson sampling}
%%%%%%%%%%%%%%%%%%%%%%%%%%%%%%%%%%%%%%%%%%%%%%%%%%%%%%%%%%%%%%%%%%%%%%%%%%%%%
As pointed out  by \citet{acharya2014poissonization}, the minimax redundancy of a class of stationary memoryless sources is equal to the minimax redundancy of the induced class on \emph{types}. Recall that the type of a sequence \citep{csiszar:korner:1981} is the sequence of frequencies of the different symbols in the sequence. As such, it is a random variables with values in the set of infinite sequences of integers with finitely many non-zero entries. 
More precisely, for $P\in\cC$, let us denote by $\tau(P^n)$ the distribution of the \emph{type} of a sequence $(X_1,\dots,X_n)\sim P^n$, that is  $\tau(P^n)$ is the probability distribution of the sequence $(N_n^j)_{j\geq 1}$. For a class $\cC$ of sources over $\cX$,  the class of probability distributions $\tau(\cC^n)$ is defined as as
\begin{eqnarray*}
\tau(\cC^n)&=&\left\{\tau(P^n),\, P\in\cC\right\}\, .
\end{eqnarray*}
As types are minimal sufficient statistics, 
\begin{displaymath}
\RED(\cC^n)=\RED(\tau(\cC^n))\quad\text{and}\quad \REG(\cC^n)=\REG(\tau(\cC^n))\, .
\end{displaymath}

Poissonization provides a simplified framework in various settings. 
Under \emph{Poisson sampling}, the message length is a random variable $N$ that is distributed according to a Poisson distribution with mean $n$ ($N\sim\cP(n)$). Picking a Poissonized sample from $P$ consist first of picking $N$ according to  $\cP(n)$ and then picking an i.i.d. sample of random length $N$ from $P$. 

 Let $\cC^{\cP(n)}$ be the Poissonized version of $\cC^n$: 
$$
\cC^{\cP(n)}=\left\{P^{\cP(n)},\, P\in\cC\right\},
$$
where, for all $P\in\cC$, all $k\geq 1$ and $x_{1:k}\in\cX^k$,
$$
P^{\cP(n)}(x_{1:k})=\PROB(N=k).P^k(x_{1:k}).
$$
By convention, the empty sequence has probability zero. For $j\geq 1$, let us denote by $N_j(n)$ the number of occurrences of symbol $j$ in a Poisson sample with size $\cP(n)$. Then $N_j(n)$ is distributed as $\cP(nP(j))$, and a very useful property of Poisson sampling is that the symbol counts $(N_j(n))_{j\geq 1}$ are independent. Let us also note that, as in the fixed-$n$ setting, the redundancy of $\cC^{\cP(n)}$ is equal to the type-redundancy \citep{acharya2014poissonization}:
\begin{equation}\label{eq:pois-type}
\RED(\cC^{\cP(n)})=\RED\left(\tau(\cC^{\cP(n)})\right)\quad \mbox{and}\quad \REG(\cC^{\cP(n)})=\REG\left(\tau(\cC^{\cP(n)})\right)\, .
\end{equation}

\begin{prop}\label{prop:pois-redun-representation}
For all class $\cC$, the Poissonized minimax redundancy satisfies
\begin{eqnarray*}
\RED(\cC^{\cP(n)})&=& \inf_{(Q_k)}\sup_{P\in\cC}\sum_{k\geq 0}\PROB(N=k)D(P^k,Q_k)
=  \sum_{k\geq 0} \PROB(N=k) \RED(\cC^{k})\, , 
\end{eqnarray*}
where the infimum is taken over sequences $(Q_k)_{k\geq 0}$, such that, for all $k\geq 0$, $Q_k$ is a probability distribution over $\cX^k$.
\end{prop}

Minimax regret and minimax regret under Poisson sampling are related by Theorem 2 from \citep{acharya2014poissonization}. \citet{2015arXiv150408070F} also show 
that $\RED(\cC^{\cP(n)})\leq 2\RED(\cC^n)$. The next proposition complements these connections.

\begin{prop}\label{prop:poisson-redundancy}
For any class $\cC$ of stationary memoryless sources with $\RED(\cC)<\infty$, 
\begin{displaymath}
\RED\left(\cC^{\cP(n-n^{2/3})}\right) +o_{\cC}(1) 
\leq \RED(\cC^n)
\leq (1-o(1))\RED\left(\cC^{\cP(n+n^{2/3})}\right)
\end{displaymath}
\end{prop}

%%%%%%%%%%%%%%%%%%%%%%%%%%%%%%%%%%%%%%%%%%%%%%%%%%%%%%%%%%%%%%%%%%%%%%%%%%%
\subsection{Minimax redundancy of envelope classes}
%%%%%%%%%%%%%%%%%%%%%%%%%%%%%%%%%%%%%%%%%%%%%%%%%%%%%%%%%%%%%%%%%%%%%%%%%%%%%%

When the source class is an envelope source class $\cC=\cC(f)$, the class of types distributions $\tau\left(\cC(f)^{\cP(n)}\right)$ under Poisson sampling 
is a class of product distributions, 
\begin{equation}\label{eq:pois-type-representation}
\tau\left(\cC(f)^{\cP(n)}\right)=\left\{\prod_{j=1}^\infty \cP(np_j) : P=(p_j)_{j\geq 1}\in\cC(f)\right\}\, .
\end{equation}
Let us define, for $\lambda\geq 0$,
\begin{eqnarray*}
\cP^\star(\lambda)&=&\left\{\cP(\mu) : \mu\leq\lambda\right\}\, ,
\end{eqnarray*}
the class of Poisson distributions with mean smaller than $\lambda$. The minimax redundancy of envelope classes under Poisson sampling 
is tightly related to the sum of minimax redundancies of classes of Poisson distributions. 
\begin{lem}\label{lem:red-pois-sum}
Let $\cC(f)$ be an envelope class. Then for all $n\geq 0$,
\begin{displaymath}
\sum_{j=\ell_f}^\infty\RED(\cP^\star(nf_j))\leq \RED\left(\cC(f)^{\cP(n)}\right)\leq \ell_f\RED(\cP^\star(n))+\sum_{j=\ell_f}^\infty \RED(\cP^\star(nf_j))\, .
\end{displaymath}
\end{lem}

The next lemma provides us with a lower bound on $\RED\left(\cP^\star(nf_j)\right)$.

\begin{lem}\label{lem:red-pois-class}
For $\lambda\geq 1$, the redundancy of $\cP^\star(\lambda)$ satisfies
\begin{eqnarray*}
\RED(\cP^\star(\lambda))&\geq & \frac{\log\lambda}{2}-5\, .
\end{eqnarray*}
\end{lem}

The first lower bound in Theorem \ref{prop:red-upper} now follows from Proposition \ref{prop:poisson-redundancy}, Lemma \ref{lem:red-pois-sum}, and Lemma \ref{lem:red-pois-class}.

%%%%%%%%%%%%%%%%%%%%%%%%%%%%%%%%%%%%%%%%%%%%%%%%%%%%%%%%%%%%%%%%%%%%%%%%%%%%%
%%%%%%%%%%%%%%%%%%%%%%%%%%%%%%%%%%%%%%%%%%%%%%%%%%%%%%%%%%%%%%%%%%%%%%%%%%%%%%
%\input{bhc-code}
%%%%%%%%%%%%%%%%%%%%%%%%%%%%%%%%%%%%%%%%%%%%%%%%%%%%%%%%%%%%%%%%%%%%%%%%%%%%%%
%!TEX root = envelope.tex

\section{The Pattern Censoring code} \label{sec:bhc-code}
%%%%%%%%%%%%%%%%%%%%%%%%%%%%%%%%%%%%%%%%%%%%%%%%%%%%%%%%%%%%%%%%%%%%%%%%%%%%%%%%
%%%%%%%%%%%%%%%%%%%%%%%%%%%%%%%%%%%%%%%%%%%%%%%%%%%%%%%%%%%%%%%%%%%%%%%%%%%%%%%%
%%%%%%%%%%%%%%%%%%%%%%%%%%%%%%%%%%%%%%%%%%%%%%%%%%%%%%%%%%%%%%%%%%%%%%%%%%%%%%%%

The \underline{P}attern \underline{C}ensoring code (\textsc{pc} code) borrows ideas from pattern coding \citep{MR2097043,garivier:2006}.
The encoder and the decoder maintain a dictionary of symbols seen so far. With each symbol is associated the number of occurrences of  the 
symbol and its rank of insertion in the dictionary.  
If the $n^{\text{th}}$ symbol in the message is the first occurrence of symbol $j$, an escape symbol is fed to the Krichevsky-Trofimov encoder, 
symbol $j$ is inserted into a dictionary with count $1$ and symbol $j$ is encoded using a general purpose prefix code for  integer (an Elias encoder). If the $n^{\text{th}}$ symbol has already 
been inserted into the dictionary, the index of its first occurrence is fed to the Krichevsky-Trofimov encoder.

We first give a bird-eye view of the \textsc{pc}  encoding of a sequence $x_{1:n}= x_1,\ldots,x_n$ of symbols from the message alphabet $\mathcal{X}=\N_+= \N\setminus \{0\}$. The message is appended with a $0$ that will serve as a termination signal. 

\vspace{6pt}
\renewcommand{\labelitemi}{---}
\begin{itemize}
\item[] \textbf{Encoding}:
\item The dictionary is initialized as follows:  $\mathcal{D}_{0} = \left\{ \langle 0, 0 \rangle \right\}$, \textit{i.e.} we start with the single (escape) symbol $0$, associated with its virtual \emph{rank of insertion}, $0$.
\item At every index $i$ corresponding to an input symbol, maintain a dictionary $\mathcal{D}_i$ containing symbol $0$ and the symbols occurring in $x_{1:i}$, along with their rank of insertion. Dictionary $\mathcal{D}_i$ thus contains $\langle 0, 0 \rangle$ and ordered pairs
$\langle j, k \rangle$ if symbol $j$ is the $k^{\text{th}}$ distinct symbol occurring in $x_{1:i}$ ($\exists m \leq i,  x_m=j, K_{m-1}=k-1, K_m=k$).
\item Create a \emph{censored} sequence $\widetilde{x}_{1:n}$ such that every symbol $x_i$ that does not belong to $\mathcal{D}_{i-1}$ 
is replaced by the special $0$ symbol and every symbol $x_i$ that belongs to $\mathcal{D}_{i-1}$ is replaced by its rank of insertion:
\[
	\widetilde{x}_i = \begin{cases}
  	k & \mbox{if } \exists k \leq K_{i-1}, \quad \langle x_i , k  \rangle \in \mathcal{D}_{i-1} \\
  	0 & \text{otherwise\,.} 
  \end{cases}
\]
%where $K_{i}$ denotes the number of distinct symbols in $x_{1:i}$. 
Note that each time $i$ such that  $\widetilde{x}_i=0$ corresponds to the discovery of a new symbol and thus to an insertion in  the dictionary: $\mathcal{D}_i=\mathcal{D}_{i-1}\cup \langle x_i, K_i\rangle$, with $K_i=K_{i-1}+1$.
\item Let $K_n$ be the number of \emph{redacted} (censored-out) input symbols (this is the number of distinct symbols in $x_{1:n}$) 
and let $i_{1:K_n}$ be their indexes. Extract the subsequence $x_{i_{1:K_n}}$ of all such redacted symbols.
\item Perform an instantaneously decodable lossless progressive encoding (in the style of \underline{M}ixture / arithmetic coding) of the censored sequence $\widetilde{x}_{1:n}$, assuming decoder side-information about past symbols. Call the resulting string $C_M$.
\item Perform an instantaneously decodable lossless encoding (in the style of \underline{E}lias / integer coding) of each redacted symbol in $x_{i_{1:K_n}}$ and of the appended $0$ individually rather than as a sequence, assuming decoder side-information about past symbols. Call the resulting string $C_E$.
\item Interleave the coded redacted symbols of $C_E$ just after each coded $0$ symbol of $C_M$, to form the overall \textsc{pc}-code.
\item The first $0$ in the sequence of redacted symbols signals termination. 
\end{itemize}
\vspace{6pt}
\begin{itemize}
\item[] \textbf{Decoding}:
\item Decode the interleaved $C_M$ and $C_E$ strings until exhaustion, as follows.
\item Decode $C_M$ to obtain $\widetilde{x}_{1:n}$ progressively.
\item When a $0$ is encountered, decode the single interleaved redacted symbol from $C_E$ to take the place of the $0$ symbol in the decoded sequence, then move back to decoding $C_M$.
\item Note that at all times the decoder knows the entire past sequence, and therefore the decoder side past side-information hypothesis is upheld.
\end{itemize}
\vspace{6pt}
The censored sequence generated by message \emph{abracadabra}  is \[
  \textrm{00010101231}\, 
\] 
while the dictionary constructed from this sequence is \[
  \left\{\langle 0 ,0\rangle,\langle a,1\rangle, \langle b,2\rangle, \langle r,3\rangle, \langle c ,4\rangle, \langle d,5\rangle\right\}
\]
We now give the details of the dictionary, the encoding of the censored sequence $\widetilde{x}_{1:n}$, and the encoding of the redacted symbols $x_{i_{1:K_n}}$. We also take additional care in guaranteeing that the overall \textsc{pc} code is instantaneously decodable. 

%% Do we initialize ?

The censored sequence $\widetilde{x}_{1:n}$ is encoded into the string $C_M$ as follows. We start by appending an extra $0$ at the end of the original censored sequence, to signal the termination of the input. We therefore in fact encode $\widetilde{x}_{1:n}0$ into $C_M$. We do this by performing a progressive arithmetic coding \citep{rissanen:langdon:1984} using coding probabilities $\widetilde{Q}_{n+1} (\widetilde{x}_{1:n}0)$ given by:
\begin{displaymath}
 \widetilde{Q}_{n+1} (\widetilde{x}_{1:n}0) = \widetilde{Q}_{n+1}(0\mid \widetilde{x}_{1:n}) \prod_{i=0}^{n-1} \widetilde{Q}_{i+1} (\widetilde{x}_{i+1}\mid \widetilde{x}_{1:i}) \, ,
\end{displaymath}
where the predictive probabilities $\widetilde{Q}_{i+1}(\cdot\mid \widetilde{x}_{1:i})$ are given by Krichevsky-Trofimov mixtures on the alphabet $\{0,1,\dots,K_i\}$,
\begin{displaymath}
 \widetilde{Q}_{i+1} \left(\widetilde{X}_{i+1} = k \mid \widetilde{X}_{1:i}=\widetilde{x}_{1:i} \right) = \frac{\widetilde{n}^k_i+\tfrac{1}{2}}{i+\tfrac{K_{i}+1}{2}}\, ,
\end{displaymath}
where, for $0\leq k\leq K_i$,
$\widetilde{n}^k_i$ is the number of occurrences of symbol $k$ in $\widetilde{x}_{1:i}$. Note that, for $0\leq k\leq K_i$,
$$
\widetilde{n}^k_i=
\begin{cases}
K_i &\mbox{if $k=0$\, ,}\\
n_i^j -1 &\mbox{if $1\leq k\leq K_i$ and $\langle j,k\rangle \in\mathcal{D}_i$\, ,}
\end{cases}
$$
where $n_i^j$ is the number of occurrences of symbol $j$ in $x_{1:i}$.
%What these coding probabilities represent, in effect, is a mixture code consisting of progressively enlarging the alphabet based on the thresholds to include symbols $\{0,1,\cdots,K_i\}$, and feeding an arithmetic coder with Krichevsky-Trofimov mixtures over this growing alphabet. Thanks to $K_i$ being determined by the data, the enlargement of the alphabet is performed online.

The subsequence $x_{i_{1:K_n}}$ of redacted symbols is encoded into the string $C_E$ as follows.  For each $i\in i_{1:K_n}$, we encode $x_i$ using Elias penultimate coding \citep{MR0373753}. Thus, if $x_i=j$ and $N_{i-1}^j=0$, the cost of inserting this new symbol in the dictionary (the Elias encoding of $j$) is $1+ \log(j\vee 1)+2\log(1+\log(j\vee 1))$. The extra $0$ appended to the message is fed to the Elias encoder. Since no other encoded redacted symbol but this one equals $0$, it unambiguously signals to the decoder that the last $0$ symbol decoded from $C_M$ is in fact the termination signal. This ensures that the overall code is instantaneously decodable, and that it  corresponds to a coding probability $Q_n$.

%The  \textsc{pc}  encoder uses an arithmetic encoder and an Elias encoder as subroutines.
%First $0$  is appended to the message $x_{1:n} \in \mathbb{N}_+^n$. Then the   \textsc{pc} encoder 
%scans the appended message as described in algorithm \ref{bhc:encoder}. 

The number of elementary operations required by dictionary 
 maintenance depends on the chosen implementation \citep[See][]{Tar83}. If a self-adjusting balanced search tree is chosen, 
 the number of elementary operations is at most proportional to the logarithm of the size of the dictionary. 
As the expected total number of symbols inserted in the dictionary coincides with the number of distinct symbols $K_n$, on average it is 
 sub-linear, the total expected computational cost of dictionary maintenance is $O(n \log n)$ if $n$ denotes the message length.   

After reading the $i^{th}$ symbol from the message, the alphabet used by the Krichevsky-Trofimov encoder is $0, \ldots, K_i$, the state of the arithmetic encoder is a function of the counts $\widetilde{n}^0_i , \widetilde{n}^1_i, \ldots \widetilde{n}^{K_i}_i$. Counts may be handled using map or dictionary data structures provided by modern programming languages. 

\begin{algorithm}[H]
\caption{\textsc{pc} encoder}
\begin{algorithmic}[1]
\label{bhc:encoder}
  \STATE Append $0$ at the end of message $x_{1:n}$
  \STATE $\mathcal{D} \leftarrow \{ \langle 0, 0 \rangle\}\qquad $   \COMMENT{\emph{Initialize dictionary}} 
  \STATE $n^0 \leftarrow 0 \qquad$  \COMMENT{\emph{Initialize counters}} 
  \STATE \text{Initialize the arithmetic encoder}
  \STATE $K \leftarrow 0 \qquad$  \COMMENT{\emph{Initialize the number of distinct symbols seen so far}} 
  \FOR{ $i=1$ \TO $\text{length}(x_{1:n}0)$ }
    \STATE $j \leftarrow x_i$
    \IF{symbol $j \in \mathcal{D}$ and $j\neq 0$}
    	\STATE $k \leftarrow $  \text{ rank of insertion of } $j$ in $\mathcal{D} \qquad$ 
     	\COMMENT{  \emph{exists $\langle j, k \rangle \in \mathcal{D}$} }    	
    	\STATE {feed mixture arithmetic encoder with } $k$
    	\STATE emit output of arithmetic encoder if any
    	\STATE $n^k \leftarrow  n^k+1 \qquad$ 	 \COMMENT{\emph{update counter}} 
    \ELSE 
        \STATE {feed mixture arithmetic encoder with escape symbol} $0$ \\
        \COMMENT{ \emph{this forces the arithmetic coder to output the encoding of the current substring} }
        \STATE emit the  output of the arithmetic encoder 
		\STATE {feed  the Elias encoder  with} $j$
		\STATE emit the output of the Elias encoder
		\STATE $n^0 \leftarrow n^0 +1 $
		\IF{$j\neq 0$}
			\STATE $K \leftarrow K +1 \qquad$   \COMMENT{\emph{increment number of distinct symbols}}
			\STATE $\mathcal{D} \leftarrow \mathcal{D} \cup \{ \langle j, K \rangle \}  $ 
            \STATE $n^K \leftarrow 1 \qquad$  \COMMENT{\emph{initialize new counter}}
        \ELSE 
        	\STATE exit 	
    	\ENDIF
    \ENDIF
  \ENDFOR
\end{algorithmic}
\end{algorithm}

From a bird-eye viewpoint, the sequence of censored symbols defines a parsing of the message into substrings of
uncensored symbols that are terminated  by $0$. Each substring is encoded by an arithmetic encoder provided with incrementally 
updated sequences of probability vectors. The arithmetic  encodings of substrings are interleaved with Elias encodings of censored symbols.

Encoding and decoding are performed in an incremental way, even though arithmetic coding may require buffering \citep{ShaZaFe06}.

%%%%%%%%%%%%%%%%%%%%%%%%%%%%%%%%%%%%%%%%%%%%%%%%%%%%%%%%%%%%%%%%%%%%%%%%%%%%%%%%%%%%%%%%%%%%%%%%%%%%%%%%%%%%%%%%%%%%%%%%%%%%%%%%
%\end{document}
%%%%%%%%%%%%%%%%%%%%%%%%%%%%%%%%%%%%%%%%%%%%%%%%%%%%%%%%%%%%%%%%%%%%%%%%%%%%
%%%%%%%%%%%%%%%%%%%%%%%%%%%%%%%%%%%%%%%%%%%%%%%%%%%%%%%%%%%%%%%%%%%%%%%%%%%%%%
%\input{code_analysis}
%%%%%%%%%%%%%%%%%%%%%%%%%%%%%%%%%%%%%%%%%%%%%%%%%%%%%%%%%%%%%%%%%%%%%%%%%%%%%%
%!TEX root = envelope.tex

%%%%%%%%%%%%%%%%%%%%%%%%%%%%%%%%%%%%%%%%%%%%%%%%%%%%%%%%%%%%%%%%%%%%%%%%%%%%
%%%%%%%%%%%%%%%%%%%%%%%%%%%%%%%%%%%%%%%%%%%%%%%%%%%%%%%%%%%%%%%%%%%%%%%%%%%%%%
\section{Redundancy of the Pattern Censoring Code}
%%%%%%%%%%%%%%%%%%%%%%%%%%%%%%%%%%%%%%%%%%%%%%%%%%%%%%%%%%%%%%%%%%%%%%%%%%%%%%
%%%%%%%%%%%%%%%%%%%%%%%%%%%%%%%%%%%%%%%%%%%%%%%%%%%%%%%%%%%%%%%%%%%%%%%%%%%%%%

\label{sec:code_analysis}

The redundancy of the \textsc{pc} code, that is the cumulative entropy risk of the coding probability with respect to the sampling probability is 
the sum of expected instantaneous redundancies. 
\begin{eqnarray*}
  D(P^n, Q_n)
  & = & \sum_{i=0}^{n-1} \EXP_{P^i} \left[ \EXP_P\left[\log\frac{P(X_{i+1})}{Q_{i+1}(X_{i+1}|X_{1:i})}\Big| X_{1:i}\right]\right]  \, . 
\end{eqnarray*}
For $i\geq 0$, the conditional expected instantaneous redundancy is given by
\begin{eqnarray*}\label{eq:inst-red}
\lefteqn{\EXP_P\left[\log\frac{P(X_{i+1})}{Q_{i+1}(X_{i+1}|X_{1:i})}\Big| X_{1:i}\right]}
\\ 
&=& \underbrace{\sum_{j\geq 1} p_j\IND_{N_i^j>0}\log\left(\frac{p_j\left(i+\frac{K_i+1}{2}\right)}{N_i^j-\frac{1}{2}}\right) }_{(\textsc{i})}
%\\
%& & \hspace{0.8cm}
+ \underbrace{\sum_{j\geq 1}p_j\IND_{N_i^j=0} \log\left(\frac{p_j\left(i+\frac{K_i+1}{2}\right)}{K_i +\frac{1}{2}}\right) }_{(\textsc{ii})}
%\\
%& &\hspace{0.8cm}
+ \underbrace{\sum_{j\geq 1}p_j\IND_{N_i^j=0}
\left(1+\log(j)+2\log\left(\log(j)+1\right)\right)}_{\textsc{(iii)}}\, .  \nonumber
\end{eqnarray*}
Terms (\textsc{i}) and \textsc{(ii)} correspond to the mixture encoding on the censored sequence while  the last term corresponds  to the Elias encoding of first occurrences of symbols. 

Let us point out right away that,  in contrast with the analyses from \citep{Bon11,bontemps2012adaptiveb,boucheron2015adaptive},  the analysis of the \textsc{pc} code does not separate the contributions of the two interleaved codes. We proceed cautiously in order to take into account the fact that redundancy accounting has to match the (unknown) envelope tail behavior.  
%% There are remarks in \citep{MR2097043} concerning this issue. 
%% Encoding the dictionary and the pattern are note treated separately. 
%% 

The conditional instantaneous redundancy may be rearranged  as follows:
{\small
\begin{eqnarray}\label{eq:inst-red-rearrangement}
\lefteqn{\EXP_P\left[\log\frac{P(X_{i+1})}{Q_{i+1}(X_{i+1}|X_{1:i})}\Big| X_{1:i}\right]}\\
&=& \underbrace{\sum_{j\geq 1}p_j\log\left(ip_j\left(1+\frac{K_i+1}{2i}\right)\right)}_{A:=} + \underbrace{\sum_{j\geq 1} p_j\IND_{N_i^j>0}\log\left(\frac{1}{N_i^j-\frac{1}{2}}\right)}_{B:=}
%\\ & &\hspace{0.8cm}
+ \underbrace{\sum_{j\geq 1}p_j\IND_{N_i^j=0}\left(1+\log\left(\frac{j}{K_i+\frac{1}{2}}\right)\right)}_{C:=} +\underbrace{2\sum_{j\geq 1}p_j\IND_{N_i^j=0}\log\left(\log(j)+1\right)}_{D:=}\, .\nonumber
\end{eqnarray}
}
Note that in term $C$ 
% $$
% \sum_{j\geq 1}p_j\IND_{N_i^j=0}\log\left(\frac{j+1}{K_i+\frac{1}{2}}\right)
% $$
the $\log(j)$ contribution of Elias encoding is balanced by the $\log(K_i +1/2)$ coming from the encoding of the $0$ symbols in the censored sequence. This turns out to be crucial and allows us to obtain adaptivity within a $\log\log n$ factor (which comes from the residual part of Elias code) instead of a $\log n$ factor. The encoding of the $0$ symbols in the censored sequence is then pivotal, this is one of the main differences with the \textsc{ac} or \textsc{etac} codes: in those previous codes, zeros were actually not encoded, and the {\KT} code was working with the counts stemming from the sequence $X_{1:i}$ instead of $\widetilde{X}_{1:i}$. Here, the {\KT} code considers $0$ as a proper symbol in the sequence $\widetilde{X}_{1:i}$, and this allows us to gain the term $\textsc{(ii)}$  in \eqref{eq:inst-red} 
% $$\sum_{j\geq 1}p_j\IND_{N_i^j=0} \log\left(\frac{p_j\left(i+\frac{K_i+1}{2}\right)}{K_i +\frac{1}{2}}\right)\, , $$ 
at the price of decreasing the other counts $N_i^j$ by one, which does not significantly affects the redundancy.

Bounding the redundancy of the \textsc{pc} code then consists in two steps. The instantaneous redundancy incurred by encoding the $(i+1)^{\text{th}}$ symbol 
 code is first 
upper bounded by an expression involving $\EXP K_i$,
the expected number of distinct symbols in a message of length $i$ as well as other occupancy counts. This should not be regarded as surprising 
if we keep in mind that redundancy and regret of pattern coding essentially depend on the profile of the message \citep{MR2024675,MR2097043,MR2095850}, that is on the occupancy counts. This first upper bound is distribution-free, it holds for all stationary memoryless sources.    

\begin{prop}\label{prop:dist-free-redundancy-bound}
For any source $P$, the instantaneous redundancy of the \textsc{pc} code satisfies
\begin{eqnarray*}
\EXP_P\left[\log\frac{P(X_{i+1})}{Q_{i+1}(X_{i+1}|X_{1:i})}\right]
	& \leq &  \kappa  \frac{\EXP K_i}{i} + \sum_{j\geq 1} p_j\PROB\{N_i^j=0\}\left(\log\left(\frac{j}{\EXP K_i}\right) + 2\log\left(\log(j)+1\right)\right) \, , 
\end{eqnarray*}
where  $\kappa\leq 19.$

\end{prop}
The second step assumes that the source probability mass function is dominated by a regularly varying envelope.

\begin{prop}
Let $P\in\cC(f)$ and assume that $\vec\nu_f(1/\cdot)\in\RV_\alpha$ with $\alpha\in[0,1[$. Then there exists $i_0\in\N$ such that for all $i\geq i_0$,
\begin{eqnarray*}
\EXP_P\left[\log\frac{P(X_{i+1})}{Q_{i+1}(X_{i+1}|X_{1:i})}\right]&\leq & c_\alpha\frac{\log\log(i)\vec\nu_f(1/i)}{i}\, ,
\end{eqnarray*}
with $c_\alpha=3\left(\Gamma(1-\alpha)+\frac{1}{1-\alpha}\right)$. 

\end{prop}

Before proceeding with the analysis of the code, 
we first state some useful comparisons between the expected occupancy counts, the number of distinct symbols in a random message, the missing mass and the measure $\nu_1$, under the source and under the envelope. We start  again with distribution-free statements, and proceed with the precise asymptotic results which are valid in the regular variation scenario. As our code is fundamentally related to occupancy counts and to the occurrences of new symbols, those will be very helpful to evaluate the contribution of each term to the redundancy.

\begin{lem}\label{lem:link_Mi0_Ki_vecnu} 
For any stationary memoryless source defined by the counting function $\vec \nu$, for all $i \geq 1$, 
\begin{enumerate}[i)]
\item \label{lem:item1} 
The expected missing mass, the expected number of singletons and the expected number of distinct symbols are connected:  \[
	\EXP M_{i,0} \leq \EXP \frac{K_{i,1}}{i} \leq \EXP \frac{K_i}{i}\, ,
\]
\item \label{lem:item2} \[
    \frac{\mathe-1}{\mathe} \vec \nu(1/i)\leq 	\EXP K_i \leq \vec \nu(1/i) + i \nu_1(0,1/i) \, 
\]
\item \label{lem:item3} \[
	1 \leq \EXP K_i \times \EXP [1/K_i  ]\leq 3 \, . 
\]

\end{enumerate}
\end{lem}
The largest value $\kappa_1$ of  $\EXP [1/K_i] \EXP K_i $ over all
 $i \in \mathbb{N}$ and over all sampling distributions is yet to be determined. It is not smaller than $9/8$: consider a geometric distribution with success probability $2/3$.

The expected number of distinct symbols under the sampling distribution is related to 
the expected number of distinct symbols under the envelope distribution. 
\begin{lem}
 If $P\in\cC(f)$, then $\vec\nu(1/i)\leq\ell_f+\vec\nu_f(1/i)$ and $\EXP K_i\leq \ell_f +\EXP_f[K_i]$.
\end{lem}
Finally, if the counting function $\vec{\nu}_f$  associated with the envelope is regularly varying in the sense of \citep{MR0216548}, the expected number of distinct symbols under the envelope distribution is simply connected with $\vec{\nu}_f$. 
\begin{lem} \citep{MR0216548,GneHanPit07}\label{lem:link_Mi0_Ki_vecnu_2}
If $\vec\nu_f(1/\cdot)\in\RV_\alpha$ with $\alpha\in[0,1[$, then 
\[
	\EXP_f[K_n]  \underset{n \to \infty}\sim   \Gamma(1-\alpha)\vec\nu_f(1/n)  \, .
\]
For all $\varepsilon>0$, there exists  $n_0\in\N$,  such that for $ n\geq n_0$,	
    \[
    n \, \nu_{1,f}[0,1/n] \leq  \frac{\alpha+\varepsilon}{1-\alpha} \vec\nu_f(1/n)\, .
\]
\end{lem}
Notice that $\Gamma(1)=1$ whereas $\lim_{\alpha \nearrow 1}\Gamma(1- \alpha) = \infty.$ 
The heavier the envelope tail, the larger the ratio between $\EXP_f K_n$ and $\vec{\nu}_f(1/n)$.

We now proceed with the analysis of the redundancy of the \textsc{pc} code. 

\begin{proof}[Proof of Theorem \ref{thm:main-compression}]
Recall expansion  \eqref{eq:inst-red-rearrangement} of  the instantaneous redundancy incurred 
by encoding the $(i+1)^{\text{th}}$ symbol, given $X_{1:i}$:
{\small
\begin{eqnarray*}
\lefteqn{\EXP_P\left[\log\frac{P(X_{i+1})}{Q_{i+1}(X_{i+1}|X_{1:i})}\Big| X_{1:i}\right]}\\
&=& \underbrace{\sum_{j\geq 1}p_j\log\left(ip_j\left(1+\frac{K_i+1}{2i}\right)\right)}_{A:=} + \underbrace{\sum_{j\geq 1} p_j\IND_{N_i^j>0}\log\left(\frac{1}{N_i^j-\frac{1}{2}}\right)}_{B:=}
%\\ & &\hspace{0.8cm}
+ \underbrace{\sum_{j\geq 1}p_j\IND_{N_i^j=0}
\left(1+\log\left(\frac{j}{K_i+\frac{1}{2}}\right)\right)}_{C:=} +\underbrace{2\sum_{j\geq 1}p_j\IND_{N_i^j=0}\log\left(\log(j)+1\right)}_{D:=}\, .\nonumber
\end{eqnarray*}
}

Starting with term $A$, 
we use the fact that for all $x\geq 0$, $\log(1+x)\leq\log(\mathe)x$ and obtain
\begin{eqnarray*}
	A & \leq & \sum_{j\geq 1}p_j\log(ip_j) +\log(\mathe)\frac{K_i+1}{2i}\, .
\end{eqnarray*}
Averaging over the first $i$ symbols $X_{1:i}$ leads to   
\begin{eqnarray*}
\EXP[A]
&\leq & \sum_{j\geq 1}p_j\log(ip_j) +\log(\mathe)\frac{\EXP[K_i]+1}{2i}\, .
\end{eqnarray*}
Moving on to term $B$, the conditional Jensen inequality entails 
\begin{eqnarray*}
\EXP[B] &=&\EXP\left[ \sum_{j\geq 1}p_j\IND_{N_i^j>0}\EXP\left[\log\frac{1}{N_i^j-\frac{1}{2}}\,\big|\, N_i^j>0\right]\right]\\
&\leq &\EXP\left[ \sum_{j\geq 1}p_j\IND_{N_i^j>0} \log\EXP\left[\frac{1}{N_i^j-\frac{1}{2}}\,\big|\, N_i^j>0\right]\right] \, . 
\end{eqnarray*}
Resorting to Lemma \ref{lem:exp-inverse-binom} in Appendix \ref{sec:conc-moment-bounds}, 
\begin{eqnarray*}
\EXP[B]&\leq & \EXP\left[\sum_{j\geq 1}p_j\IND_{N_i^j>0}\log\left(\frac{1}{ip_j}\left(1+\frac{9}{ip_j}\right)\right)\right]\\
&\leq & -\sum_{j\geq 1}p_j\PROB(N_i^j>0)\log(ip_j)+\log(\mathe)\frac{9\EXP[K_i]}{i}\, .
\end{eqnarray*}
Hence, 
\begin{eqnarray*}
\EXP[A+B]&\leq & \frac{1}{i}\sum_{j\geq 1} ip_j\mathe^{-ip_j}\log(ip_j)+\log(\mathe)\frac{10\EXP[K_i]}{i}\, .
\end{eqnarray*}
For $p_j\leq 1/i$ or equivalently $j> \vec\nu(1/i)$, $\log(ip_j) \leq 0$. As for all $x\geq 1$, $x \exp(-x) \log(x) \leq \log(\mathe)/\mathe$, 
\begin{eqnarray*}
\frac{1}{i}\sum_{j\geq 1} ip_j\mathe^{-ip_j} \log(ip_j) & \leq & \frac{1}{i}\sum_{j\leq \vec{\nu}(1/i)} ip_j \mathe^{-ip_j} \log(ip_j)  \\
&\leq  & \frac{\vec\nu(1/i)}{i}  \frac{\log (\mathe)}{\mathe}\,\leq \, \frac{\log(\mathe)}{\mathe -1}\frac{\EXP K_i}{i} \, ,
% \\ &\leq & \frac{\log (\mathe)}{\mathe} \frac{\vec\nu_f(1/i)+\ell_f}{i}\, .
\end{eqnarray*}
where the last inequality comes from Lemma \ref{lem:link_Mi0_Ki_vecnu}, \ref{lem:item2}). We get
\begin{equation}\label{eq:sumAB}
	\EXP[A+B] \leq  \left(\frac{\log (\mathe)}{\mathe-1}  +10\log(\mathe)\right)\frac{\EXP[K_i]}{i}\, .
\end{equation}
We turn towards term $C$. As $K_i+1/2\geq \frac{1}{2}(K_i+1)$, we have
\begin{eqnarray*}
C&\leq & 2M_{i,0} +\sum_{j\geq 1}p_j\IND_{N_i^j=0}\log\left(\frac{j}{K_i+1}\right)\\
&=& 2M_{i,0} + \sum_{j\geq 1}p_j\IND_{N_i^j=0}\log\left(\frac{j}{\sum_{\ell\neq j}\IND_{N^{\ell}_{i}>0}+1}\right)\, .
\end{eqnarray*}
The collection of random variables $(N_i^j)_{j\geq 1}$ is  \emph{negatively associated} (See Appendix \ref{sec:negative_association}). 
Hence, for all $j\geq 1$, $N_i^j$  and $\sum_{\ell \neq j} N^\ell_j$ are negatively associated, 
\begin{eqnarray*}
\EXP\left[\IND_{N_i^j=0}\log\left(\frac{j}{\sum_{\ell\neq j}\IND_{N_i^{\ell}>0}+1}\right)\right]&\leq & \PROB(N_i^j=0)\EXP\left[\log\left(\frac{j}{\sum_{\ell\neq j}\IND_{N_i^{\ell}>0}+1}\right)\right]\, .
\end{eqnarray*}
Resorting to Jensen's inequality and noticing that $\sum_{\ell\neq j}\IND_{N_i^{\ell}>0}+1\geq K_i$, we get
\begin{eqnarray*}
\EXP[C]
&\leq & 2\EXP[M_{i,0}] + \sum_{j\geq 1}p_j\PROB(N_i^j=0)\log\EXP\left[\frac{j}{K_i}\right]\, .
\end{eqnarray*}
Using Lemma \ref{lem:link_Mi0_Ki_vecnu}, \ref{lem:item3}), we obtain
\begin{eqnarray}\label{eq:C}
\EXP[C]&\leq & \left(2+\log 3\right)\EXP[M_{i,0}] +\sum_{j\geq 1}p_j\PROB(N_i^j=0)\log\left(\frac{j}{\EXP K_i}\right)\, .
\end{eqnarray}
Note that Proposition \ref{prop:dist-free-redundancy-bound} follows from \eqref{eq:sumAB} and \eqref{eq:C}, combined with Lemma \ref{lem:link_Mi0_Ki_vecnu}, \ref{lem:item1}. 

The next step consists in assuming that the source belongs to an envelope class $\cC(f)$. 

\begin{lem}\label{lem:sum-envelope}
Let $f$ be an envelope function and assume $P\in\cC(f)$. Then
\begin{enumerate}[i)]
\item 
\begin{eqnarray*}
\sum_{j\geq 1}p_j\PROB(N_i^j=0)\log\left(\frac{j}{\EXP K_i}\right)&\leq & 2\frac{\EXP_f K_i +\ell_f}{i}+\sum_{j\geq \vec\nu_f(1/i)} f_j\log\left(\frac{j}{\vec\nu_f(1/i)}\right)\, .
\end{eqnarray*}
\item 
\begin{eqnarray*}
\EXP[D]&\leq & 2\sum_{j\geq \vec\nu_f(1/i)}f_j\log\left(\log(j)+1\right)+2\frac{\EXP_f K_i +\ell_f}{i}\log\left(\log\left(\vec\nu_f(1/i)\right)+1\right)\, .
\end{eqnarray*}
\end{enumerate}
\end{lem}

Finally, in a last step, the regular assumption on the envelope allows us to evaluate each of the terms involved in the instantaneous redundancy.

\begin{lem}\label{lem:eval-sum-reg-var}
Let $P\in\cC(f)$ and assume that $\vec\nu_f(1/\cdot)\in\RV_\alpha$, for $\alpha\in[0,1[$. For all $\varepsilon>0$, there exists $i_0\in\N$ such that, for all $i\geq i_0$ 
\begin{enumerate}[i)]
\item 
\begin{eqnarray*}
\sum_{j\geq \vec\nu_f(1/i)} f_j\log \frac{j}{\vec\nu_f (1/i)}
     &\leq & \frac{3-\alpha}{(1-\alpha)^2}\;  \frac{\vec\nu_f(1/i)}{i} \, .
\end{eqnarray*}
\item 
\begin{eqnarray*}
\EXP[D]&\leq & 2(1+\varepsilon)\left(\Gamma(1-\alpha)+\frac{1}{1-\alpha}\right)\frac{\log\log(i)\vec\nu_f(1/i)}{i}\, .
\end{eqnarray*}
\end{enumerate}
\end{lem}

As the bound on $\EXP [D]$ in Lemma \ref{lem:eval-sum-reg-var} is much larger than $\vec\nu_f(1/i)/i$, we finally obtain that there exists $i_0\in\N$ that may depend on the envelope $f$ such that for all $i\geq i_0$,
\begin{eqnarray*}
\EXP_P\left[\log\frac{P(X_{i+1})}{Q(X_{i+1}|X_{1:i})}\right]&\leq & c_\alpha\frac{\log\log(i)\vec\nu_f(1/i)}{i}\, ,
\end{eqnarray*}
with $c_\alpha=3\left(\Gamma(1-\alpha)+\frac{1}{1-\alpha}\right)$. 

Hence, for $n\geq i_0$, we obtain
\begin{eqnarray*}
\RED(Q_n,\cC(f)^n)
&\leq & C(i_0,f)+c_\alpha\log\log(n)\sum_{i=1}^n \frac{\vec\nu_f(1/i)}{i}\, ,
\end{eqnarray*}
which establishes the upper bound in Theorem \ref{thm:main-compression}.

\end{proof}
%%%%%%%%%%%%%%%%%%%%%%%%%%%%%%%%%%%%%%%%%%%%%%%%%%%%%%%%%%%%%%%%%%%%%%%%%%%%%%
%%%%%%%%%%%%%%%%%%%%%%%%%%%%%%%%%%%%%%%%%%%%%%%%%%%%%%%%%%%%%%%%%%%%%%%%%%%%%%
%\input{proofs}
%%%%%%%%%%%%%%%%%%%%%%%%%%%%%%%%%%%%%%%%%%%%%%%%%%%%%%%%%%%%%%%%%%%%%%%%%%%%%%
%!TEX root = envelope.tex
%%%%%%%%%%%%%%%%%%%%%%%%%%%%%%%%%%%%%%%%%%%%%%%%%%%%%%%%%%%%%%%%%%%%%%%%%%%
\section{Proofs}
%%%%%%%%%%%%%%%%%%%%%%%%%%%%%%%%%%%%%%%%%%%%%%%%%%%%%%%%%%%%%%%%%%%%%%%%%%%%%%%%%%
\label{sec:proofs}

\subsection{Proof of Theorem \ref{prop:red-upper}}

\begin{enumerate}[i)]
\item As announced, the upper bound follows from bounds of \citep{acharya2014poissonization} on the minimax regret defined as
\begin{eqnarray*}
\REG(\cC(f)^n)&=&\inf_{Q_n}\sup_{P\in\cC(f)} \sup_{x_{1:n}\in\cX^n}\log\frac{P^n(x_{1:n})}{Q_n(x_{1:n})}\, .
\end{eqnarray*}
Clearly, $\RED\left(\cC(f)^n\right)\leq \REG\left(\cC(f)^n\right)$. Now Theorem 14 of \citep{acharya2014poissonization} states that
\begin{eqnarray*}
\REG\left(\cC(f)^n\right)&\leq & 1+\REG\left(\cC(f)^{\cP(n)}\right)\,=\,1+\REG\left(\tau(\cC(f)^{\cP(n)})\right)\, ,
\end{eqnarray*}
Using \eqref{eq:pois-type-representation} and the fact that for $j< \ell_f$, we still have the crude bound $p_j\leq 1$,
\begin{eqnarray*}
\tau(\cC(f)^{\cP(n)})&\subset & \left(\prod_{j=1}^{\ell_f-1}\cP^\star(n)\right)\prod_{j=\ell_f}^\infty \cP^\star(nf_j)\, .
\end{eqnarray*}
As enlarging the class can not reduce the regret, and as, for two classes $\cC_1, \cC_2$, $\REG(\cC_1\times\cC_2)=\REG(\cC_1)+\REG(\cC_2)$, we have
\begin{eqnarray*}
\REG\left(\tau(\cC(f)^{\cP(n)})\right)&\leq & \ell_f\REG\left(\cP^\star(n)\right) +\sum_{j\geq \ell_f}\REG\left(\cP^\star(nf_j)\right)\, .
\end{eqnarray*}
Lemma 17 of \citep{acharya2014poissonization} provide us with the following bounds: if $\lambda\leq 1$,
\begin{displaymath}
\REG\left(\cP^\star(\lambda)\right)=\log\left(2-\mathe^{-\lambda}\right)\leq \log(\mathe)\lambda\, ,
\end{displaymath}
and, if $\lambda>1$, 
\begin{displaymath}
\REG\left(\cP^\star(\lambda)\right)\leq \log\left(\sqrt{\frac{2\lambda}{\pi}}+2\right)\, .
\end{displaymath}
For aesthetic purposes (getting the common $\log(\mathe)$ constant in front of each terms), we find it convenient to notice that the bound above can be very slightly improved to
\begin{displaymath}
\REG\left(\cP^\star(\lambda)\right)\leq \log\left(\sqrt{\frac{2\lambda}{\pi}}+\frac{3}{2}\right)\, .
\end{displaymath}
As $\frac{3}{2}+\sqrt{\frac{2}{\pi}}\leq \mathe$,
\begin{eqnarray*}
\REG\left(\cC(f)^n\right)&\leq &  O(\ell_f\log n)
+\sum_{j, f_j\geq 1/n} \log\left(\mathe\sqrt{nf_j}\right)+ \log(\mathe)\sum_{j, f_j\leq 1/n}nf_j\, .
\end{eqnarray*}
Now, using the integral representation and integrating by parts gives
\begin{eqnarray*}
\sum_{f_j\geq 1/n}\log\left(\mathe\sqrt{nf_j}\right)&=&\log(\mathe)\vec\nu_f(1/n)+ \int_{1/n}^1\frac{\log(nx)}{2}\nu_f(\mathd x)\\
&=&\log(\mathe)\left(\vec\nu_f(1/n) +\int_{1/n}^1\frac{\vec\nu_f(x)}{2x}\mathd x\right)\, .
\end{eqnarray*}
Also, we may write $\sum_{j, f_j\leq 1/n}f_j=\nu_{1,f}[0,1/n]$, which gives the desired upper bound.
\item Let $m=n-n^{2/3}$. Thanks to Proposition \ref{prop:poisson-redundancy}, for $n$ large enough
\begin{eqnarray*}
\RED\left(\cC(f)^n\right)&\geq & \RED\left(\cC(f)^{\cP(m)}\right)-1\, .
\end{eqnarray*}
Now, by Lemmas \ref{lem:red-pois-sum} and \ref{lem:red-pois-class},
\begin{displaymath}
\RED\left(\cC(f)^{\cP(m)}\right)\geq  \sum_{j=\ell_f}^\infty \RED\left(\cP^\star(mf_j)\right)\geq  \sum_{j,\, f_j\geq 1/m}\RED\left(\cP^\star(mf_j)\right)\geq \sum_{j,\, f_j\geq 1/m}\left(\frac{\log(mf_j)}{2}-5\right)\, .
\end{displaymath}
Using the integral representation and integrating by parts,
\begin{eqnarray*}
\sum_{j,\, f_j\geq 1/m}\frac{\log(mf_j)}{2}&=&\frac{\log(\mathe)}{2}\int_{1/m}^1 \ln(mx)\nu_f(\mathd x)\\
&=&\frac{\log(\mathe)}{2}\left(\left[-\vec\nu_f(x)\ln(mx)\right]_{1/m}^1 +\int_{1/m}^1 \frac{\vec\nu_f(x)}{x}\mathd x\right)\\
&=& \log(\mathe)\int_1^m\frac{\vec\nu_f(1/t)}{2t}\mathd t\, .
\end{eqnarray*}
We obtain
\begin{eqnarray}\label{eq:red-lower-first}
\RED\left(\cC(f)^n\right)&\geq & \log(\mathe)\int_1^m\frac{\vec\nu_f(1/t)}{2t}\mathd t -5\vec\nu_f(1/m)-1\, .
\end{eqnarray}
Note that the lower bound \eqref{eq:red-lower-first} might be irrelevant. For instance, when $\vec\nu_f(1/\cdot)\in\RV_\alpha$ with $\alpha>\log(\mathe)/10$, the right-hand side becomes negative. However, an order-optimal lower bound in the case of heavy-tailed envelopes (that is, when $\alpha>0$) was established in \citep{boucheron2015adaptive}: the expected redundancy of a class $\cC(f)^n$ is lower-bounded by the expected number of distinct symbols in a sample of size $n$ drawn according to the envelope distribution $(f_j)_{j\geq 1}$. More precisely, for all envelope function $f$, there exists a constant $c_f\geq 0$ such that, for all $n\geq 1$, $\RED\left(\cC(f)^n\right)\geq \EXP_f[K_n]-c_f$.
\end{enumerate}
%\end{proof}

%%%%%%%%%%%%%%%%%%%%%%%%%%%%%%%%%%%%%%%%%%%%%%%%%%%%%%%%%%%%%%%%%%%%%%%%%%%%%%%
\subsection{Proofs of Section \ref{sec-minmaxred}}

%%%%%%%%%%%%%%%%%%%%%%%%%%%%%%%%%%%%%%%%%%%%%%%%%%%%%%%%%%%%%%%%%%%%%%%%%%%%%%
\begin{proof}[Proof of Proposition \ref{prop:pois-redun-representation}]
We have
$$
\RED(\cC^{\cP(n)})=\inf_{Q\in\cM_1(\cX^*)}\sup_{P\in\cC}D(P^{\cP(n)},Q).
$$
Let $P\in\cC$ and $Q\in\cM_1(\cX^*)$, with $\cX^* = \cup_{k\geq 1} \cX^k$. The distribution $Q$ can be written as $Q=\sum_{k\geq 1} q(k)Q_k$, where $(q(k))_{k\geq 1}$ is a probability distribution over $\N$, and, for each $k\geq 1$, $Q_k$ is a distribution over $\cC^k$. Hence
\begin{eqnarray}\label{eq:pois-KL}
D(P^{\cP(n)},Q)&=& \sum_{k\geq 1}\PROB(N=k)\sum_{x\in\cC^k}P^k(x)\log\frac{\PROB(N=k)P^k(x)}{q(k)Q_k}\nonumber\\
&=& D(\cP(n),q)+\sum_{k\geq 1}\PROB(N=k)D(P^k,Q_k)\, .
\end{eqnarray}
Maximizing in $P$ and minimizing in $(q(k))_{k\geq 1}$ and $(Q_k)_{k\geq 0}$, we get
\begin{eqnarray*}
\RED(\cC^{\cP(n)})
%&=&\inf_{(Q_k)}\inf_{(q(k))} \left(D(\cP(n),q)+\sup_{P\in\cC}\sum_{k\geq 0}\PROB(N=k)D(P^k,Q_k)\right)\\
&=& \inf_{(q(k))}D(\cP(n),q) + \inf_{(Q_k)}\sup_{P\in\cC}\sum_{k\geq 1}\PROB(N=k)D(P^k,Q_k)\, .
\end{eqnarray*}
The first term is equal to zero for $q=\cP(n)$, implying that the distribution $Q$ which achieves the minimax redundancy is also a Poisson mixture. Hence
\begin{eqnarray*}
\RED(\cC^{\cP(n)})&=& \inf_{(Q_k)}\sup_{P\in\cC}\sum_{k\geq 1}\PROB(N=k)D(P^k,Q_k)\, .
\end{eqnarray*}
\end{proof}

%%%%%%%%%%%%%%%%%%%%%%%%%%%%%%%%%%%%%%%%%%%%%%%%%%%%%%%%%%%%%%%%%%%%%%%%%%%%%%%

\begin{proof}[Proof of Proposition \ref{prop:poisson-redundancy}]
We start with the lower bound on $\RED(\cC^{n})$. Let $m=n-n^{2/3}$, and let $M$ be a Poisson random variable with mean $m$. By Proposition \ref{prop:pois-redun-representation},
\begin{eqnarray*}
\RED(\cC^{\cP(m)})&=& \inf_{(Q_k)}\sup_{P\in\cC}\sum_{k\geq 0}\PROB(M=k)D(P^k,Q_k)\\
&\leq & \inf_{(Q_k)}\sum_{k\geq 0}\PROB(M=k)\sup_{P\in\cC}D(P^k,Q_k)\\
%&=&\sum_{k\geq 0}\PROB(M=k)\RED(\cC^k)\, .
\end{eqnarray*}
Using the fact that the sequence $(\RED(\cC^k))_{k\geq 0}$ is increasing and sub-additive (see Proposition \ref{prop:red-properties}), we have
\begin{eqnarray*}
\RED(\cC^{\cP(m)})&\leq & \RED(\cC^n) +\sum_{k>n} \PROB(M=k)\left(\RED(\cC^k)-\RED(\cC^n)\right)\\
&\leq & \RED(\cC^n) +\sum_{k>n} \PROB(M=k)\RED(\cC^{k-n})\\
&\leq & \RED(\cC^n) +\RED(\cC)\sum_{k>n} \PROB(M=k)(k-n)\, .
\end{eqnarray*}
Resorting to Lemma \ref{lem:poisson-concentration} in Appendix \ref{sec:conc-moment-bounds}, we have
\begin{align*}
 \sum_{k>n} \PROB(M=k)(k-n) & = \EXP\left[(M-n)\IND_{\{M>n\}}\right]
= \int_0^\infty \PROB\left(M-m>t+n^{2/3}\right)\mathd t & \\ 
& \leq  \int_0^n \mathe^{-\frac{n^{4/3}}{2(m+n^{2/3})}}\mathd t + \int_n^\infty \mathe^{-\frac{t^2}{6t}}\mathd t & \\
& \leq  n\mathe^{-n^{1/3}/2} + 6\mathe^{-n/6} \underset{n \to \infty}\longrightarrow 0\, . 
\end{align*}
This establishes the lower bound on $\RED(\cC^n)$ in Proposition \ref{prop:poisson-redundancy}. 

Let us now establish the upper bound. Let now $m=n+n^{2/3}$ and $M$ be a Poisson random variable with mean $m$. Using the Bayesian representation of the minimax redundancy (see Proposition \ref{prop:red-properties}), we have
\begin{eqnarray*}
\RED(\cC^{\cP(m)})&=&\sup_{\pi\in\cM_1(\cC)}\inf_{Q\in\cM_1(\cX^*)}\int D(P^{\cP(m)},Q)\mathd\pi(P)\, .
\end{eqnarray*}
Fix $\pi\in\cM_1(\cC)$. Resorting to Equation \eqref{eq:pois-KL} in the proof of Proposition \ref{prop:pois-redun-representation}, we have
\begin{eqnarray*}
& &\inf_{Q\in\cM_1(\cX^*)}\int D(P^{\cP(m)},Q)\mathd\pi(P)\\
& &\hspace{1.6cm} = \inf_{(Q_k),(q(k))} \int \left(D(\cP(m),q) +\sum_{k\geq 0}\PROB(M=k)D(P^k,Q_k)\right) \mathd\pi(P)\\
& &\hspace{1.6cm} = \inf_{(Q_k)}\int \sum_{k\geq 0}\PROB(M=k)D(P^k,Q_k) \mathd\pi(P)\\
& &\hspace{1.6cm} = \sum_{k\geq 0}\PROB(M=k)\inf_{Q_k}\int D(P^k,Q_k) \mathd\pi(P)\, .
\end{eqnarray*}
We claim that the sequence $\left(\inf_{Q_k}\int D(P^k,Q_k) \mathd\pi(P)\right)_{k\geq 0}$ is increasing with respect to $k$. Indeed, let $k\geq 0$ and let $Q_{k+1}\in \cM_1(\cX^k)$. Denote by $Q_{k+1}^{(k)}$ its restriction to the first $k$ symbols. Then, for all $P\in\cM_1(\cX)$, $D(Q_{k+1},P^{k+1})\geq D(P^k,Q_{k+1}^{(k)})$. Hence for all $Q_{k+1}$ there exist $Q_k\in\cM_1(\cX^k)$ such that $\int D(P^k,Q_k) \mathd\pi(P)\leq \int D(P^{k+1},Q_{k+1}) \mathd\pi(P)$, which gives the desired result. We get 
\begin{eqnarray*}
\RED(\cC^{\cP(m)})&\geq &\sup_{\pi}\sum_{k\geq n}\PROB(M=k)\inf_{Q_k}\int D(P^k,Q_k) \mathd\pi(P)\\
%&\geq &\PROB(M\geq n)\, \sup_{\pi}\inf_{Q_n}\int D(P^n,Q_n) \mathd\pi(P)\\
&\geq & \PROB(M\geq n) \RED(\cC^n)\, .
\end{eqnarray*}
Now, using again Lemma \ref{lem:poisson-concentration}, we have
\begin{eqnarray*}
\PROB(M\geq n)&\geq & 1-\exp\left(-\frac{n^{4/3}}{2m}\right)\;\underset{n \to \infty}\longrightarrow 1\, ,
\end{eqnarray*}
which concludes the proof.
\end{proof}

%%%%%%%%%%%%%%%%%%%%%%%%%%%%%%%%%%%%%%%%%%%%%%%%%%%%%%%%%%%%%%%%%%%%%%%%%%

\begin{proof}[Proof of Lemma \ref{lem:red-pois-class}]
Using the Bayesian representation of the minimax redundancy, we have
\begin{eqnarray*}
\RED(\cP^\star(\lambda))&=& \sup_{\pi\in\cM_1([0,\lambda])}\int D(\cP(\mu),\cP_\pi)\mathd\pi(\mu)\, ,
\end{eqnarray*}
where $\displaystyle{\cP_\pi=\int \cP(\mu)\mathd\pi(\mu)}$. In particular, taking $\pi$ equal to the uniform distribution over $[0,\lambda]$, we get
\begin{eqnarray*}
\RED(\cP^\star(\lambda))&\geq & \int_0^\lambda\frac{1}{\lambda}\sum_{k\geq 0}\PROB(\cP(\mu)=k)\log \frac{\PROB(\cP(\mu)=k)}{\PROB(\cP_\pi=k)}\mathd\mu\, .
\end{eqnarray*}
We have
\begin{eqnarray*}
\PROB(\cP_\pi =k)&=&\frac{1}{\lambda} \int_0^\lambda\frac{\mathe^{-\mu}\mu^k}{k!}\mathd\mu \,=\, \frac{\PROB(\cP(\lambda)>k)}{\lambda}\,\leq\,\frac{1}{\lambda}\, .
\end{eqnarray*}
Hence
\begin{eqnarray*}
\RED(\cP^\star(\lambda))&\geq &\log\lambda -\frac{1}{\lambda}\int_0^\lambda H(\cP(\mu))\mathd\mu \, .
\end{eqnarray*}
Using Stirling's bound $k!\leq \mathe^{1/12k}\left(\frac{k}{\mathe}\right)^k\sqrt{2\pi k}$, we have, for all $\mu\in[0,\lambda]$,
\begin{eqnarray*}
\frac{H(\cP(\mu))}{\log(\mathe)}&=&\mu - \mu\ln\mu +\sum_{k\geq 0}\frac{\mathe^{-\mu}
\mu^k}{k!}\ln(k!)\\
&\leq & \mu - \mu\ln\mu +\sum_{k\geq 1}\frac{\mathe^{-\mu}
\mu^k}{k!}\left(k\ln k -k +\frac{\ln(2\pi k)}{2}+\frac{1}{12k}\right)\\
&\leq & \sum_{k\geq 1}\frac{\mathe^{-\mu}
\mu^k}{k!}\left(k\ln k +\frac{\ln(2\pi k)}{2}\right)-\mu\ln\mu +\frac{1}{12}\, .
\end{eqnarray*}
We use Jensen's inequality to obtain
\begin{eqnarray*}
\sum_{k\geq 1}\frac{\mathe^{-\mu}
\mu^k}{k!}k\ln k&=&\mu \sum_{k\geq 0}\frac{\mathe^{-\mu}
\mu^k}{k!}\ln(k+1)\,\leq \, \mu\ln(1+\mu)\, ,
\end{eqnarray*}
and
\begin{eqnarray*}
\sum_{k\geq 1}\frac{\mathe^{-\mu}
\mu^k}{k!}\ln k&\leq &(1-\mathe^{-\mu})\ln\left(\frac{\mu}{1-\mathe^{-\mu}}\right)\,\leq\, \ln\mu+\frac{1}{\mathe}\, ,
\end{eqnarray*}
where the last inequality is due to the fact that the function $x\mapsto x\ln x$ is larger than $-1/\mathe$ for all $x\geq 0$. We get
\begin{eqnarray*}
\frac{H(\cP(\mu))}{\log(\mathe)}&\leq & \frac{\ln\mu}{2} +\mu\ln\left(1+\frac{1}{\mu}\right)+\frac{\ln(2\pi)}{2}+\frac{1}{2\mathe}+\frac{1}{12}\\
&\leq & \frac{\ln\mu}{2} +3\, ,
\end{eqnarray*}
which is smaller than $3$ for $\mu\leq 1$. Hence
\begin{eqnarray*}
\RED(\cP^\star(\lambda))&\geq &\log\lambda -3\log(\mathe)-\frac{1}{\lambda}\int_1^\lambda \frac{\log\mu}{2}\mathd\mu \\
&\geq & \frac{\log\lambda}{2}-5\, .
\end{eqnarray*}
\end{proof}

%%%%%%%%%%%%%%%%%%%%%%%%%%%%%%%%%%%%%%%%%%%%%%%%%%%%%%%%%%%%%%%%%%%%%%%%%%%%%%%%

\subsection{Proofs of Section \ref{sec:code_analysis}}

%%%%%%%%%%%%%%%%%%%%%%%%%%%%%%%%%%%%%%%%%%%%%%%%%%%%%%%%%%%%%%%%%%%%%%%%%%%%%%%

\begin{proof}[Proof of Lemma \ref{lem:link_Mi0_Ki_vecnu}]
\begin{enumerate}[i)]
\item 
Observe that
\begin{displaymath}
\EXP M_{n,0}=\int_0^1x(1-x)^n\nu(\mathd x)\leq \frac{1}{n}\int_0^1 nx(1-x)^{n-1}\nu(\mathd x)=\frac{\EXP K_{n,1}}{n}\leq \frac{\EXP K_n}{n}\, .
\end{displaymath}
The relation between $\EXP K_n$ and $\EXP_f K_n$ is easily obtained by noticing that the function $x\mapsto 1-(1-x)^n$ is increasing on $[0,1]$, which gives
\begin{eqnarray*}
\EXP K_n&=&\sum_{j\geq 1}\left(1-(1-p_j)^n\right)\, \leq \, \ell_f +\sum_{j\geq 1}\left(1-(1-f_j)^n\right)\,=\, \ell_f+\EXP_f K_n\, .
\end{eqnarray*}
\item
The expected number of distinct symbols in a message of length $i$ is never much less that the number of symbols with probability larger than $1/i$: 
\begin{displaymath}
\EXP K_i \geq \sum_{j\leq\vec\nu (1/i)}(1-\left(1-p_j)^i\right) \geq  \frac{\mathe - 1}{\mathe}\vec\nu(1/i)\, .
\end{displaymath}
The upper bound is obtained in a similar way 
\begin{eqnarray*}
\EXP K_i & =  &  \sum_{j\geq 1} \left(1 - (1-p_j)^i\right)	\\
&  \leq & \vec \nu(1/i) + \sum_{j> \vec \nu(1/i)} \left(1 - (1-p_j)^i\right) \\
& \leq  & \vec \nu(1/i) + \sum_{j> \vec \nu(1/i)} i p_j \\
& = & \vec \nu(1/i) + i \nu_1(0,1/i) \, 
\end{eqnarray*}
where the second inequality comes from $1-x \leq (1- x/i)^i$ for $0\leq x \leq 1$  and $i \geq 1$. 

\item 
The quantity $\EXP 1/K_i$  satisfies a quadratic inequality:
\begin{eqnarray*}
\EXP\left[\frac{1}{K_i}\right] - \frac{1}{\EXP K_i}  
	 & = & \EXP \left[\frac{\EXP[K_i]- K_i}{\sqrt{\EXP K_i}}\times \frac{1}{K_i \sqrt{\EXP K_i}} \right] \\ 
	 & \leq &  \EXP \left[\frac{(\EXP[K_i]- K_i)^2}{{\EXP K_i}}\right]^{1/2} \times \EXP\left[\frac{1}{K_i^2 \EXP K_i}\right]^{1/2} \\ 
	 & \leq & \frac{(\EXP K_{i,1})^{1/2}}{\EXP K_i} \times \EXP\left[\frac{1}{K_i}\right]^{1/2}  
\end{eqnarray*}
where the first	inequality follows by invoking Cauchy-Schwarz inequality,  the first factor in the penultimate line is upper bounded by $1$ thanks to the fact that  $\operatorname{var}(K_i)\leq \EXP K_{i,1}$ \citep{ben2014concentration}, while the second factor 
is crudely upper bounded using the fact that $K_i \geq 1$  entails $1/K_i^2 \leq 1/K_i$. Solving the  inequality leads to 
\begin{eqnarray*}
	\EXP\left[\frac{1}{K_i}\right] 
	& \leq  & 
	\frac{1}{\EXP K_i}\left(1 + \frac{\EXP K_{i,1}}{2\EXP K_i}  +
	\left(\frac{\EXP K_{i,1}}{\EXP K_i}\right)^{1/2} \sqrt{1 +\frac{\EXP K_{i,1}}{4\EXP K_i} }\right) \\
	& \leq & \frac{1}{\EXP K_i} \left(1+ \frac{\EXP K_{i,1}}{\EXP K_i} + \left(\frac{\EXP K_{i,1}}{\EXP K_i}\right)^{1/2}\right) \\ 
	& \leq & \frac{1}{\EXP K_i} \left(1+ 2\left(\frac{\EXP K_{i,1}}{\EXP K_i}\right)^{1/2}\right)\\
	& \leq & \frac{3}{\EXP K_i} \, . 
\end{eqnarray*}

\end{enumerate}
\end{proof}

%%%%%%%%%%%%%%%%%%%%%%%%%%%%%%%%%%%%%%%%%%%%%%%%%%%%%%%%%%%%%%%%%%%%%%%%%

\begin{proof}[Proof of Lemma \ref{lem:sum-envelope}]
\begin{enumerate}[i)]
\item 
We have
\begin{eqnarray*}
\sum_{j\geq 1}p_j\PROB(N_i^j=0)\log\left(\frac{j}{\EXP[K_i]}\right)& =  & \sum_{j\geq 1}p_j\PROB(N_i^j=0)\log\left(\frac{j}{\EXP_f[K_i]}\right)+ \EXP[M_{i,0}]\log\frac{\EXP_f[K_i]}{\EXP[K_i]} \\
& \leq & \sum_{j\geq 1}p_j\PROB(N_i^j=0)\log\left(\frac{j}{\EXP_f[K_i]}\right)+ \frac{\EXP K_i}{i} \log \frac{\EXP_f[K_i]}{\EXP[K_i]} \\
& \leq & \sum_{j\geq 1}p_j\PROB(N_i^j=0)\log\left(\frac{j}{\EXP_f[K_i]}\right)+  \frac{\EXP_f[K_i]}{2i} \, ,
\end{eqnarray*}
where the first inequality follows from  $\EXP[M_{i,0}]\leq \EXP[K_i]/i$ (see Lemma \ref{lem:link_Mi0_Ki_vecnu}),
and the second inequality follows from the fact that $x \mapsto x \log \frac{y}{x}$ achieves its maximum $y/2$ at $x=y/2$. 

By Lemma \ref{lem:link_Mi0_Ki_vecnu}, \ref{lem:item2}, we have $\EXP_f K_i\geq  \frac{\mathe - 1}{\mathe}\vec\nu_f(1/i)$. Hence
\begin{eqnarray*}
\sum_{j\geq 1}p_j\PROB(N_i^j=0)\log\left(\frac{j}{\EXP_f[K_i]}\right)&\leq & \log\frac{\mathe}{\mathe-1} \, \EXP[M_{i,0}] + \sum_{j\geq 1}p_j\log\left(\frac{j}{\vec\nu_f(1/i)}\right)\, .
\end{eqnarray*}
Now, for $j<\vec\nu_f(1/i)$, the summands are negative and we simply omit them to get
\begin{eqnarray*}
\sum_{j\geq 1}p_j\log\left(\frac{j}{\vec\nu_f(1/i)}\right)
&\leq &\sum_{j\geq\vec\nu_f(1/i)}f_{j}\log\left(\frac{j}{\vec\nu_f(1/i)}\right)\, ,
\end{eqnarray*}
where we used $p_{j}\leq f_j$, provided that $j\geq \ell_f$. 
\item Decomposing the sum according to $\vec\nu_f(1/i)$ and using the fact that $\EXP[M_{i,0}]\leq (\EXP_f K_i +\ell_f)/i$ readily yield the desired bound on $\EXP[D]$.
\end{enumerate}
\end{proof}

%%%%%%%%%%%%%%%%%%%%%%%%%%%%%%%%%%%%%%%%%%%%%%%%%%%%%%%%%%%%%%%%%%%%%%%%%%%%%%

\begin{proof}[Proof of Lemma \ref{lem:eval-sum-reg-var}]
\begin{enumerate}[i)]
\item 
Using the integral representation, we have
\begin{eqnarray*}
\sum_{j\geq \vec\nu_f(1/i)} f_j\log\left(\frac{j}{\vec\nu_f(1/i)}\right)
&= &\int_0^{1/i} x\log\frac{\vec\nu_f(x)}{\vec\nu_f(1/i)}\nu_f(\mathd x)\, . 
\end{eqnarray*}
By the regular variation assumption on $\vec\nu_f(1/\cdot)$, the Potter-Drees Inequality (see section \ref{sec:reg-var-tools}) implies that: for all $\varepsilon,\delta>0$, $\exists i_0\in\N$ such that, for all $i\geq i_0$, for all $x\in(0,1/i]$,
\begin{eqnarray*}
\frac{\vec\nu_f(x)}{\vec\nu_f(1/i)}&\leq & \left(\frac{1}{xi}\right)^\alpha +\varepsilon\left(\frac{1}{xi}\right)^{\alpha+\delta}\, .
\end{eqnarray*}
Taking crudely $\varepsilon=1$, $\delta=1-\alpha$ and bounding $\alpha$ by $1$, we obtain that for $i$ large enough and $x\in(0,1/i]$,
\begin{eqnarray*}
\frac{\vec\nu_f(x)}{\vec\nu_f(1/i)}&\leq & \frac{2}{xi}\, .
\end{eqnarray*}
Hence
\begin{eqnarray*}
\int_0^{1/i} x\log\frac{\vec\nu_f(x)}{\vec\nu_f(1/i)}\nu_f(\mathd x) &\leq & \nu_{1,f}[0,1/i] +\log(\mathe)\int_0^{1/i}x\ln\left(\frac{1}{xi}\right)\nu_f(\mathd x)\, .
\end{eqnarray*}
By Fubini's Theorem,
\begin{eqnarray*}
\int_0^{1/i}x\ln\left(\frac{1}{xi}\right)\nu_f(\mathd x)
&=&\int_0^{1/i}x\int_1^{\frac{1}{xi}}\frac{1}{t}\mathd t\,\nu_f(\mathd x)\\
&=& \int_1^\infty\frac{1}{t}\int_0^{\frac{1}{ti}}x\nu_f(\mathd x)\mathd t \, =\, \int_i^\infty \frac{\nu_{1,f}[0,1/t]}{t}\mathd t\, .
\end{eqnarray*}
Now, Lemma \ref{lem:link_Mi0_Ki_vecnu_2} implies that, for all $\varepsilon>0$, for $i$ large enough and for all $t\geq i$, $\nu_{1,f}[0,1/t]\leq \frac{(\alpha +\varepsilon)\vec\nu_f(1/t)}{(1-\alpha)t}$, and by Karamata's Theorem (see section \ref{sec:reg-var-tools}),
\begin{eqnarray*}
\int_i^\infty \frac{\vec\nu_f(1/t)}{t^2}\mathd t &\sim &\frac{\vec\nu_f(1/i)}{(1-\alpha)i}\, .
\end{eqnarray*}
Hence, for $i$ large enough,
\begin{eqnarray*}
\sum_{j\geq \vec\nu_f(1/i)} f_j\log\left(\frac{j}{\vec\nu_f(1/i)}\right)&\leq & \frac{3-\alpha}{(1-\alpha)^2}\times  \frac{\vec\nu_f(1/i)}{i}\, .
\end{eqnarray*}
\item 
By Lemma \ref{lem:sum-envelope}, 
\begin{eqnarray}
\label{eq:bound-D}
\EXP[D]&\leq & 2\sum_{j\geq \vec\nu_f(1/i)}f_j\log\left(\log(j)+1\right)+2\frac{\EXP_f K_i +\ell_f}{i}\log\left(\log\left(\vec\nu_f(1/i)\right)+1\right)\, .
\end{eqnarray}
As $\vec\nu_f(x)\ll 1/x$ when $x\to 0$ (which is true as soon as the support is infinite, see \citep{GneHanPit07}), for $i$ large enough, we have $\log(\vec\nu_f(x))+1\leq \log(1/x)$ for all $x\in]0,1/i]$. Hence, the first term in the right-hand side of \eqref{eq:bound-D} can be controlled as follows,
\begin{eqnarray*}
\sum_{j\geq \vec\nu_f(1/i)}f_j\log\left(\log(j)+1\right) &=&\int_0^{1/i}x\log\left(\log(\vec\nu_f(x))+1\right)\nu_f(\mathd x)\\ 
&\leq & \int_0^{1/i}x\log\log\left(\frac{1}{x}\right)\nu_f(\mathd x)\\
&=& \left[-\vec\nu_f(x)x\log\log\left(\frac{1}{x}\right)\right]_0^{1/i} +\int_0^{1/i}\left(\log\log\left(\frac{1}{x}\right)+\frac{\log(\mathe)}{\ln x}\right)\vec\nu_f(x)\mathd x\\
&\leq &  \int_0^{1/i}\log\log\left(\frac{1}{x}\right)\vec\nu_f(x)\mathd x\,=\, \int_i^\infty \frac{\log\log(t)\vec\nu_f(1/t)}{t^2}\mathd t\, .
\end{eqnarray*} 
We used that, as $\alpha<1$, the limit of $\vec\nu_f(x)x\log\log(1/x)$ as $x\to 0$ is equal to $0$. Now, the function $t\mapsto\frac{\log\log(t)\vec\nu_f(1/t)}{t^2}$ belongs to $\RV_{\alpha-2}$. Hence, by Karamata's Theorem (see section \ref{sec:reg-var-tools}),
\begin{eqnarray*}
\int_i^\infty \frac{\log\log(t)\vec\nu_f(1/t)}{t^2}\mathd t &\underset{i\to\infty}\sim &\frac{\vec\nu_f(1/i)}{(1-\alpha)i}\log\log(i)\, .
\end{eqnarray*}
As for the second term in the right-hand side of \eqref{eq:bound-D}, we use $\vec\nu_f(1/i)\ll i$ and $\EXP_f K_i \sim \Gamma(1-\alpha)\vec\nu_f(1/i)$ and obtain that, for all $\varepsilon>0$, there exists $i_0\in\N$ such that, for all $i\geq i_0$ 
\begin{eqnarray*}
\EXP[D]&\leq & 2(1+\varepsilon)\left(\Gamma(1-\alpha)+\frac{1}{1-\alpha}\right)\frac{\log\log(i)\vec\nu_f(1/i)}{i}\, .
\end{eqnarray*}
\end{enumerate}
\end{proof}

%%%%%%%%%%%%%%%%%%%%%%%%%%%%%%%%%%%%%%%%%%%%%%%%%%%%%%%%%%%%%%%%%%%%%%%%%%%%%%
%%%%%%%%%%%%%%%%%%%%%%%%%%%%%%%%%%%%%%%%%%%%%%%%%%%%%%%%%%%%%%%%%%%%%%%%%%%%%%%%%

\bibliographystyle{plainnat}

%\input{appendix}
%%%%%%%%%%%%%%%%%%%%%%%%%%%%%%%%%%%%%%%%%%%%%%%%%%%%%%%%%%%%%%%%%%%%%%%%%%%%%%
%!TEX root = envelope.tex
%%%%%%%%%%%%%%%%%%%%%%%%%%%%%%%%%%%%%%%%%%%%%%%%%%%%%%%%%%%%%%%%%%%%%%%%%%%
\appendices

%%%%%%%%%%%%%%%%%%%%%%%%%%%%%%%%%%%%%%%%%%%%%%%%%%%%%%%%%%%%%%%%%%%%%%%%%%%%%%%%%%

\section{Regular variation}
\label{sec:reg-var-tools}

The determination of bounds on minimax redundancies  and the analysis of the performance  of the \textsc{pc} code for classes with regularly varying envelopes
rely at some point on classical results from regular variation theory. For the sake of self-reference, these results are recalled here \citep[See][for proofs and refinements]{BinGolTeu89,HaaFei06}.

\begin{thm}\textsc{(Karamata's integration Theorem)}
  \label{thm:karamata}
Let $f\in\RV_\alpha$ and assume that there exists $t_0>0$ such that $f$ is positive and locally bounded on $[t_0,+\infty[$. 
\begin{enumerate}[(i)]
\item If $\alpha > -1$, then
\begin{displaymath}
\int_{t_0}^t f(s)\mathd s \underset{t\to +\infty}\sim \frac{tf(t)}{\alpha +1}\, .
\end{displaymath}
\item If $\alpha<-1$, or $\alpha=-1$ and $\int_0^\infty f(s)\mathd s<\infty$, then
\begin{displaymath}
\frac{tf(t)}{\int_{t}^{+\infty} f(s)\mathd s} \underset{t\to +\infty}\sim {-\alpha -1}\, .
\end{displaymath}
\end{enumerate}
\end{thm}

\begin{thm}\textsc{(Potter-Drees inequalities.)}
  \label{thm:potter:rv}
If $f \in \textsc{rv}_\alpha$, then for all  $\delta>0$, there exists $t_0=t_0(\alpha)$, 
such that for all $t,x\colon  \min(t,tx)>t_0$,
\begin{displaymath}
 (1-\delta) x^\alpha\min\big(x^\delta,x^{-\delta} \big)  \leq \frac{f(tx)}{f(t)} \leq (1+\delta) x^\alpha\max\big(x^\delta,x^{-\delta} \big)\, .
\end{displaymath}
\end{thm}

In order to relate Karlin's infinite urn scheme setting and the setting of \citep{boucheron2015adaptive} in Appendix \ref{sec:connections_between_countin}, we need to compute asymptotic inverses of regularly varying functions.  This is done using the notion of De Bruijn conjugacy.
%===============================================================================
\begin{thm}[\textsc{de bruijn conjugacy}]\citep[Proposition 1.5.15]{BinGolTeu89} \label{the:debruijn}
  Let $\ell\in \textsc{rv}(0)$, then there exists a function $\ell^\#\in \textsc{rv}(0)$ such that $\ell^\#(x)\ell(x\ell^\#(x))\rightarrow 1$ 
  and $\ell(x)\ell^\#(x\ell(x))\rightarrow 1$ as $x \to \infty$. Any function satisfying these two relations is asymptotically equivalent to $\ell^\#$. The functions $(\ell,\ell^\#)$ are said to form a pair of De Bruijn conjugates.
\end{thm}

If $f = t^\gamma \ell(t)$ where $\ell$  is slowly varying, then $\ell^{1/\gamma}$ is also slowly varying and the function 
$y \mapsto y^{1/\gamma}\big( \ell^{1/\gamma}\big)^\#\big(y^{1/\gamma}\big)$ is an asymptotic inverse of $f$. This entails that for $\alpha>0,$ 
$f \in \textsc{rv}_ \alpha$,  any asymptotic inverse of $f$ is regularly varying with regular variation index $1/ \alpha$. 

While exploring source classes defined by slowly varying counting functions, a refined notion of regular variation due to de Haan may prove relevant. 

\begin{dfn}\label{def:dehaan} A slowly varying function $\ell$ has the extended regular variation property if there exists a slowly varying function $\ell_0$ such that for all $x>0$ \[
	\lim_{t \to \infty} \frac{\ell(tx)-\ell(t)}{\ell_0(t)} = \log (x) \, ,
\]	
	this is summarized by  $\ell \in  \Pi_{\ell_0}$. 
\end{dfn}
Slowly varying functions satisfying the extended regular variation property form a proper subset of $\textsc{rv}_0$. For example,  $g(t)= \lfloor \log(t) \rfloor$ and 
$g(t)= \lfloor \log(t) \rfloor^2$
are slowly varying but do not satisfy the extended regular variation property, while 
$g(t)= \lfloor \log(t)^2 \rfloor \in \Pi_{2 \log}$.

\section{Infinite urn schemes and regular variation} % (fold)
\label{sec:infinite_urn_sch}

% section infinite_urn_sch (end)
The next theorem connects $\EXP K_n $ and the regular variation properties of the envelope
(characterized by the regular variation of $\vec\nu$). For much more on the asymptotic behavior of occupancy counts, see \citet*{MR0216548,GneHanPit07}.

\begin{thm}[\cite{MR0216548,GneHanPit07}]\label{thm:asympt-occup}
Assume that there exists a slowly varying function $\ell$ such that $\vec\nu(1/n)= n^\alpha\ell(n)$, with $\alpha\in[0,1]$ and $\ell$ slowly varying at infinity.
\begin{enumerate}[i)]
\item If $0<\alpha<1$, 
$$
\EXP K_n \sim \Gamma(1-\alpha)\vec\nu(1/n)\, ,
\quad\mbox{and}\quad\nu_{1}[0,1/n]\sim \frac{\alpha\vec\nu(1/n)}{(1-\alpha)n}\,.
$$
\item If $\alpha=1$, 
$$
\EXP K_n \sim \EXP K_{n,1} \sim n\nu_{1}[0,1/n]\sim n\ell_1(n)\quad \text{where} 
\quad \ell_1(n)=\int_n^\infty \frac{\ell(t)}{t}\mathd t\gg\ell(n) \, . 
$$
\item If $\alpha=0$, 
$$
\EXP K_n \sim \ell(n)\, ,\quad
\mbox{and}\quad\nu_{1}[0,1/n]\ll \frac{\ell(n)}{n}\,.
$$
If furthermore,  $\vec\nu(1/\cdot) \in \Pi_{\ell_0}$, then $\nu_1[0,x] \sim x \ell_0(1/x)$ and $\ell(x)\sim\int_1^{x}u^{-1}\ell_0(u)\mathd u$. 
\end{enumerate}
\end{thm}

\section{Connections between counting function and tail quantile function} % (fold)
\label{sec:connections_between_countin}

% section connections_between_countin (end)
In \citep{bontemps2012adaptiveb,boucheron2015adaptive}, censoring methods and minimax redundancy rates for envelope classes are described 
using the envelope tail quantile function (or a smothed version of it) $U(t) =  \inf\{ x : \overline{F}(x) \leq \frac{1}{t}\}$ (where $\overline{F}$ denotes the envelope survival function). In order to clarify the connection between the performance of the \textsc{pc} code and the performance of the codes described in these two references, we relate the regular variation properties of the counting function $\vec{\nu}$ and the regular variation properties of the tail quantile function $U$. 

For the sake of simplicity we assume that the probability mass	function $(p_j)_{j\geq 1}$ is decreasing and  that no $p_j$ is null. This assumption entails 
$	\vec\nu(p_j) = j $ for $j\geq 1$.

For $t>1$, \[
	U(t) = \inf\big\{ j : \textstyle{\sum_{k>j}} p_k \leq \tfrac{1}{t}		\big\} =  \inf\big\{ j : \nu_1(0,p_j) \leq \tfrac{1}{t}		\big\} \, .
\]
This entails that $U$ may be defined from $\vec\nu$ and $\nu_1$,
\[
	U(t) = 1+ \vec{\nu}\left( \sup \left\{ x :  \nu_1(0,x) \leq \tfrac{1}{t}\right\}\right)  \, . 
\]
Note that \(
{1}/{ \sup \left\{ x :  \nu_1(0,x) \leq \tfrac{1}{t}\right\}} = \inf \left\{ y :  {1}/{\nu_1(0,1/y)} \geq t\right\} \, ,
\) so 
letting $h(y)$ be a shorthand for $1/\nu_1(0, 1/y)$, $U(t) = 1 +\vec \nu(1/h^{\leftarrow}(t))$.

\begin{lem}\label{lem:nu2U}
	Assume that there exists a slowly varying function $\ell$ such that $\vec\nu(x) \sim x^{-\alpha}\ell(1/x)$, with $\alpha\in[0,1)$ and 
	that the probability mass	function $(p_j)_{j\geq 1}$ is decreasing and  that no $p_j$ is null. Let $\gamma= \alpha/(1- \alpha)$.   	
	 Then,  
	\begin{enumerate}[i)]
		\item $U$ is asymptotically equivalent with a regularly varying function with index $\gamma.$ 
		\item If $\alpha \in (0,1)$, let  
	$\widetilde{\ell}$ be a shorthand for $(1/\ell^{1+\gamma})^\#$. Then  \[
			U(t) \underset{t \to \infty}{\sim} \gamma^\gamma t^{\gamma} \widetilde{\ell}\left(t^{1+\gamma}\right) \, . 
		\]
		\item If $\alpha=0$ and $\ell \in \Pi_{\ell_0}$, let $\widetilde{\ell}_0$ be a shorthand for $\left({1}/{\ell_0}\right)^\#$.
		Then $U$ has the extended regular variation property, $U \in  \Pi_{\widetilde{\ell}_0}$  and \[
			U(t) \underset{t \to \infty}{\sim} \ell \left( t \widetilde{\ell}_0(t) \right) \, . 
		\]	
		\item If $\alpha=0$  and the distribution  is discrete log-concave, then 
		$U(t) \underset{t \to \infty}{\sim}  \vec \nu(1/t)$. 
	\end{enumerate}
\end{lem}
\begin{proof}

If $\vec{\nu}(1/\cdot)$ has either positive regular variation index or has both the slow variation and the extended regular
variation property, then  $h^{\leftarrow}$ 
is the generalized inverse of the regularly varying function $1/\nu_1(0, 1/\cdot)$ (See Theorem \ref{thm:asympt-occup} above). 
As such, $h^{\leftarrow}$ is regularly varying, and so is $U.$

	If $h(y)= 1/\nu_1(0, 1/\cdot) \in \textsc{rv}_{1-\alpha}$ for $\alpha \in (0,1)$ , then its generalized inverse is regularly varying with index $1+\gamma= 1/(1-\alpha)$ and, from Propositions 1.5.14 and  1.5.15 from \citep{BinGolTeu89},   
	\[
		h^{\leftarrow} (t) \underset{t \to \infty}{\sim} {\gamma}^{1+\gamma} t^{1+\gamma} \widetilde{\ell}\left(t^{1+\gamma} \right)\, . 
	\]
	and 
\begin{eqnarray*}
	U(t) & \underset{t \to \infty}{\sim} &  \gamma^{\gamma} t^{\gamma} 
	\widetilde{\ell}\left(t^{1+\gamma} \right)^{\gamma/(1+\gamma)}
	\ell \big(t^{1+\gamma} \widetilde{\ell}\left(t^{1+\gamma} \right) \big)\\
	& \underset{t \to \infty}{\sim} & \gamma^{\gamma} t^{\gamma} \widetilde{\ell}\left(t^{1+\gamma} \right)  \, . 
\end{eqnarray*}

If $\vec \nu(1/\cdot) \in \Pi_{\ell_0}$ where $\ell_0$  is slowly varying, then by Theorem \ref{thm:asympt-occup} and
conjugacy arguments \[
	U(t) \sim  \ell \big( t \widetilde{\ell}_0(t) \big) \, . 
\]

If the distribution is discrete long-concave (which is equivalent to $(p_{k+1}/p_{k})$ is non increasing)
the counting function $\vec \nu$ is readily verified to be slowly varying. The fact that $U$ is slowly varying (but does not satisfy the extended regular variation property) is well known \citep{And70}. 

For all $k\geq 1$, $p_{k+1} \leq \nu_1(0, p_k) \leq p_{k+1} /(1-p_{k+1}/p_k)$. If $p_j \geq \tfrac{1}{t} > p_{j+1}$, then $j \leq U(t) \leq j+k$  where $k\leq 2+ \log(p_j/(p_j-p_{j+1})/\log(p_j/p_{j+1})$ which is bounded. This entails that $\vec \nu(1/t) \leq U(t)$ and $\lim_{t \to \infty} U(t)/\vec \nu(1/t)= 1.$ 
\end{proof}

In \citep{boucheron2015adaptive}, the \textsc{etac} code escapes the $n+1^{\text{th}}$ symbol $X_{n+1}$, if $X_{n+1} \geq M_n$ where 
\begin{equation}\label{def:mn}
M_n = \min \left( n, \left\{k\;:\;X_{k,n} \leq k \right\} \right), 
\end{equation}
while $(X_{k,n})_{k \leq n }$ is the non-decreasing rearrangement of $X_1, \ldots, X_n$. The random threshold $M_n$ is concentrated around 
$m_n$ where $m_n=m(n)$  and $m(t)$ is defined for $t\geq 1$ as the solution of equation $U(t/x)=x.$ In \citep{boucheron2015adaptive}, it is proved that 
the function $m$ inherits the regular variation properties of $U$. Namely, if $U(t)=\gamma^\gamma t^\gamma \widetilde{\ell}(t^{1+\gamma})$ where $\widetilde{\ell}= (1/\ell^{1+\gamma})^\#$ and $\ell \in \textsc{rv}_0$, $m$ is regularly varying with index $\gamma/(\gamma+1)= \alpha$,
\begin{displaymath}
 m(t) \underset{t \to \infty}{\sim} \gamma t^{\alpha} {\ell}(t) \underset{t \to \infty}{\sim} \frac{\alpha}{1-\alpha } t^{\alpha} {\ell}(t)\underset{t \to \infty}{\sim} \frac{\alpha}{1-\alpha } \vec\nu(1/t)\, .
\end{displaymath}
When $\ell$  is slowly varying, the connexion between $\vec \nu(1/\cdot)$  and $m$ is more subtle. The function $m$ is the reciprocal of a De Bruijn conjugate of $U$. Hence by Lemma \ref{lem:nu2U},  in order to have $m(t) \underset{t \to \infty}{\sim}   \ell(t)$, we need to have to have 
\[
 	\frac{\ell\left(\frac{t}{\ell(t)} \widetilde{\ell}_0\left(\frac{t}{\ell(t)}\right)\right)}{\ell(t)} \underset{t \to \infty}{\sim} 1 \, . 
 \] 
We are not aware of any meaningful characterization of this property.

\section{Negative association} % (fold)
\label{sec:negative_association}

% section negative_association (end)
When handling finite or infinite urn models, negative association arguments are a source of elegant moment or tail inequalities for occupancy counts. They complement Poissonization arguments.  

\begin{dfn}[\textsc{negative association}]
Real-valued random variables  $Z_1,\dots, Z_K$ are said to be negatively associated if, for any two disjoint subsets $A$ and $B$ of $\{1,\dots,K\}$, and any two real-valued functions $f:\mathbb{R}^{|A|}\mapsto \mathbb{R}$ and $g:\mathbb{R}^{|B|}\mapsto \mathbb{R}$ that are both either coordinate-wise non-increasing or coordinate-wise non-decreasing, we have:
$$\EXP\left[f(Z_A).g(Z_B)\right] \leq \EXP\left[f(Z_A)\right].\EXP\left[g(Z_B)\right] \, .$$
\end{dfn}

%In particular,  sums of negatively associated variables can only do better than sums of independent variables.

\begin{thm}\citep{dubhashi:ranjan:1998}
For each $n\in \mathbb{N}$,  the occupancy scores $(N^{j}_n)_{j\geq 1}$ are negatively associated.
\end{thm}

Monotonic functions of negatively associated variables are  negatively associated. Hence, 
 the variables $(\IND_{\{N^j_{n}>0\}})_{j\geq 1}$ (respectively $(\IND_{\{N^j_{n}=0\}})_{j\geq 1}$) are negatively associated as increasing (respectively decreasing) functions of $(N^j_{n})_{j \geq 1}$. 

\section{Concentration and moment bounds}
\label{sec:conc-moment-bounds}

Poisson distributions satisfy Bennett and thus Bernstein inequalities \citep[See][Chapter 2]{boluma13}.

\begin{lem}\label{lem:poisson-concentration}
For any Poisson distributed random variable $N$, for all $t>0$,
\begin{eqnarray*}
\PROB\left(N\geq \EXP N +t\right)&\leq & \exp\left(-\frac{t^2}{2(\EXP N +t/3)}\right)\, ,\\
\PROB\left(N\leq \EXP N -t\right)&\leq & \exp\left(-\frac{t^2}{2\EXP N}\right)\, .
\end{eqnarray*}
\end{lem}

Bounds on expected inverse binomial random variables can be found in the literature \citep[See][and  references therein]{GyKhKrWa06,Arl09}. The next results were developed for the purpose of this paper. 

\begin{lem}\label{lem:exp-inverse-binom}
Let $N$ be a binomial random variable with parameters $n$ and $p$. 
\begin{eqnarray*}
\EXP\left[\frac{1}{N-\frac{1}{2}}\,\big|\, N>0\right]&\leq & \frac{1}{\EXP N}\left(1 + \frac{9}{\EXP N} \right) \, . 
\end{eqnarray*}

\end{lem}

\begin{proof}[Proof of Lemma \ref{lem:exp-inverse-binom}]
Using the fact that, for $k\geq 1$, 
\begin{displaymath}
\frac{1}{k-\frac{1}{2}}=\frac{1}{k+1}\left(1+\frac{3}{2k-1}\right)\leq \frac{1}{k+1}+\frac{9}{(k+1)(k+2)}\, ,
\end{displaymath}
we have
\begin{eqnarray*}
\EXP\left[\frac{1}{N-\frac{1}{2}}\,\big|\, N>0\right]&\leq &
\frac{1}{1-(1-p)^n}\sum_{k=1}^{n} {n\choose k}p^k(1-p)^{n-k}\left(\frac{1}{k+1}+\frac{9}{(k+1)(k+2)}\right)\\
&= & \frac{1}{1-(1-p)^n}\left(\frac{\PROB\left(\mathcal{B}(n+1,p)\geq 2\right)}{(n+1)p}+\frac{9\PROB\left(\mathcal{B}(n+2,p)\geq 3\right)}{p^2(n+1)(n+2)}\right)\\
&\leq & \frac{1}{np}+\frac{9}{(np)^2}\, .
\end{eqnarray*}
\end{proof}

%%%%%%%%%%%%%%%%%%%%%%%%%%%%%%%%%%%%%%%%%%%%%%%%%%%%%%%%%%%%%%%%%%%%%%%%%%%%%%%%%%%%%%%%%%%%%%%%%%%%%%%%%%%%%%%%%%%%%%%%%%%%%%
%%%%%%%%%%%%%%%%%%%%%%%%%%%%%%%%%%%%%%%%%%%%%%%%%%%%%%%%%%%%%%%%%%%%%%%%%%%%%%%%%%%%%%%%%%%%%%%%%%%%%%%%%%%%%%%%%%%%%%%%%%%%%%
%%%%%%%%%%%%%%%%%%%%%%%%%%%%%%%%%%%%%%%%%%%%%%%%%%%%%%%%%%%%%%%%%%%%%%%%%%%%%%%%%%%%%%%%%%%%%%%%%%%%%%%%%%%%%%%%%%%%%%%%%%%%%%
%%%%%%%%%%%%%%%%%%%%%%%%%%%%%% THE END !!!!!!!!! 								     %%%%%%%%%%%%%%%%%%%%%%%%%%%%%%%%%%%%%%%%%
%%%%%%%%%%%%%%%%%%%%%%%%%%%%%%%%%%%%%%%%%%%%%%%%%%%%%%%%%%%%%%%%%%%%%%%%%%%%%%%%%%%%%%%%%%%%%%%%%%%%%%%%%%%%%%%%%%%%%%%%%%%%%%
%%%%%%%%%%%%%%%%%%%%%%%%%%%%%%%%%%%%%%%%%%%%%%%%%%%%%%%%%%%%%%%%%%%%%%%%%%%%%%%%%%%%%%%%%%%%%%%%%%%%%%%%%%%%%%%%%%%%%%%%%%%%%%
%%%%%%%%%%%%%%%%%%%%%%%%%%%%%%%%%%%%%%%%%%%%%%%%%%%%%%%%%%%%%%%%%%%%%%%%%%%%%%%%%%%%%%%%%%%%%%%%%%%%%%%%%%%%%%%%%%%%%%%%%%%%%%

\end{document}